\documentclass[sigconf, nonacm, pdfa]{acmart}
\AtBeginDocument{%
  }

\newcommand{\flowsc}{\textsf{FlowSC}\xspace}
\newcommand{\bp}{\textsf{BipartitePlus}\xspace}
\newcommand{\fastest}{\textsf{Fastest}\xspace}
\newcommand{\asc}{\textsf{ASC}\xspace}
\newcommand{\learnsc}{\textsf{LearnSC}\xspace}
\newcommand{\neursc}{\textsf{NeurSC}\xspace}
\newcommand{\alley}{\textsf{Alley}\xspace}
\newcommand{\alleytpi}{\textsf{AlleyTPI}\xspace}

\newcommand{\lss}{\textsf{LSS}\xspace}

\newcommand{\Yeast}{\textsf{Yeast}\xspace}
\newcommand{\Hprd}{\textsf{HPRD}\xspace}
\newcommand{\Human}{\textsf{Human}\xspace}
\newcommand{\Wordnet}{\textsf{WordNet}\xspace}
\newcommand{\Dblp}{\textsf{DBLP}\xspace}
\newcommand{\Youtube}{\textsf{Youtube}\xspace}
\newcommand{\EU}{\textsf{Eu2005}\xspace}
\newcommand{\Patents}{\textsf{Patents}\xspace}
\newcommand{\Twitter}{\textsf{Twitter}\xspace}
\newcommand{\Friendster}{\textsf{Friendster}\xspace}

\newcommand{\qerror}{\textsf{q-error}\xspace}

\usepackage[a-2b]{pdfx}
\usepackage{graphics}
\usepackage{booktabs}
\usepackage{datetime}
\usepackage[linesnumbered,ruled,vlined]{algorithm2e}
\usepackage{thmtools, thm-restate}
\usepackage{mathtools}
\usepackage{enumitem}
\usepackage[nameinlink]{cleveref}
\usepackage{caption}
\usepackage{subcaption}
\usepackage{multirow}
\usepackage{setspace}
\usepackage{lipsum}
\usepackage{multicol}
\usepackage{amsmath}
\usepackage{multirow, makecell, booktabs, diagbox}
\usepackage{cleveref}
\usepackage[dvipsnames, svgnames]{xcolor}

\newtheorem{theorem}{\textbf{Theorem}}[section]
\newtheorem{theorem*}{Theorem}
\newtheorem*{definition*}{Definition}
\newtheorem{definition}[theorem]{Definition}
\newtheorem{lemma}[theorem]{Lemma}
\newtheorem{example}[theorem]{Example}
\newtheorem*{prf*}{Proof}
\begin{document}

\title{Algebraic Subgraph Counting}

\author{Qiuyu Guo}
\affiliation{%
  \institution{University of New South Wales}
  \city{Sydney}
  \country{Australia}
}
\email{qiuyu.guo@unsw.edu.au}

\author{Jianye Yang}
\affiliation{%
  \institution{Guangzhou University \\
  PengCheng Laboratory}
  \city{Guangzhou}
  \country{China}
}
\authornote{Corresponding author}
\email{jyyang@gzhu.edu.cn}

\author{Wenjie Zhang}
\affiliation{
  \institution{University of New South Wales}
  \city{Sydney}
  \country{Australia}
}
\email{wenjie.zhang@unsw.edu.au }

\author{Hanchen Wang}
\affiliation{%
  \institution{University of Technology Sydney}
  \city{Sydney}
  \country{Australia}
}
\email{hanchen.wang@uts.edu.au}

\author{Ying Zhang}
\affiliation{%
  \institution{Zhejiang Gongshang University}
  \city{Hangzhou}
  \country{China}
}
\email{ying.zhang@zjgsu.edu.cn}

\author{Xuemin Lin}
\affiliation{%
  \institution{The Chinese University of Hong Kong, Shenzhen}
  \city{Shenzhen}
  \country{China}
}
\email{xuemin.lin@gmail.com}

\renewcommand{\shortauthors}{Guo et al.}



\begin{abstract}
Subgraph isomorphism counting is a fundamental problem in graph analytics, which aims to find the number of subgraph isomorphisms of a query graph in a data graph. The candidate tree-based framework provides a promising foundation for subgraph counting tasks, offering a unified counting paradigm that can be extended beyond tree patterns.
However, supporting subgraph isomorphism within this framework remains challenging, as it requires handling both the non-tree edge constraint and the injective mapping constraint. Although existing solutions employ sampling or learning techniques to address these constraints in this framework, they still either suffer from inherent sampling failures or rely heavily on supervision. In this paper, we propose \asc, an algebraic subgraph counting approach built on the candidate tree-based counting framework. In our method, the non-tree edge constraint is directly incorporated into the candidate tree-based counting process through a matrix-based computation method, enabling subgraph homomorphism counting with high accuracy in polynomial time. Based on the resulting subgraph homomorphism count, we further apply a local sampling method to address the injective mapping constraint, thereby obtaining the final subgraph isomorphism count. Extensive experiments show that \asc can achieve substantially better and more stable performance over the baselines across various datasets, while scaling to billion-edge graphs. Most impressively, as a non-learning method, \asc can even achieve more than an order of magnitude higher average accuracy than the state-of-the-art learning-based method \flowsc with similar efficiency. This paper is the full version of the work accepted at SIGMOD 2027. The code is available at \url{https://github.com/EricaGuoQiuyu/AlgebraicSubgraphCounting}.
\end{abstract}




\maketitle

\section{Introduction}
Given a query graph $Q$ and a data graph $G$, the problem of subgraph counting is to find the number of subgraph isomorphisms of $Q$ in $G$. It is a fundamental problem with a wide range of applications, such as query optimization~\cite{querylanguages,sparql,bonifati2020analytical, deeds2023safebound}, biological network analysis~\cite{proteinnetworks, evograph, protein-protein, protein-structures}, and fraud detection~\cite{FundementalProblems, qiu2018real}. 
Despite its importance, subgraph counting is computationally challenging due to its \#p-hardness~\cite{Fastest,flowsc}.
Exact counting via subgraph matching \cite{li2025subgraph, bee, gup, choi2023bice, circinus, VEQ, bsx, rapidmatch} is often intractable due to the inherent complexity of enumeration. Even counting a 32-clique query in a 32-clique graph can exceed a two-hour time limit in our experiments. To broaden applicability, recent efforts have focused on approximate counting \cite{Alley, CSET, Fastest, flowsc}. In this paper, we focus on approximate subgraph counting.

\vspace{1mm}
\noindent \textbf{Existing Solutions and Limitations}.
Current solutions for approximate subgraph counting can generally be classified into three categories, including summarization-based, sampling-based, and learning-based methods. 
Summarization-based methods \cite{BoundSketch,CSET,SumRDF,color} adopt a decomposition-aggregation paradigm, decomposing the query graph into substructures and aggregating their individual counts to estimate the total count. However, their reliance on the independence assumption between substructures often leads to unsatisfactory accuracy in real-world graphs.
Sampling-based methods \cite{IMPR, MOTIVO, Alley, Fastest} estimate counts via match frequencies in sampled subgraphs. While sampling can achieve high accuracy for simple queries on small graphs, they suffer from inherent sampling failures as the sample space expands. To address this, \fastest \cite{Fastest} employs advanced filtering and candidate tree-based sampling to prune the search space and improve sampling success rates. Nonetheless, eliminating sampling failures remains challenging for complex queries on large graphs.
Learning-based methods \cite{LSS, NeurSC, flowsc, Learnsc} predict counts by regressing from learned graph features. Earlier works \cite{NeurSC, LSS} struggle to capture the relationship between the graph features and counts or face training instabilities. The state-of-the-art \flowsc \cite{flowsc} mitigates these issues by simulating data flow within a candidate-tree framework and using a pretraining scheme. However, the performance of these techniques is heavily reliant on the labeled data, which is often impractical for hard query workloads where obtaining labels itself is computationally prohibitive.

\vspace{1mm}
\noindent \textbf{Matrix-based Pattern Counting.}
The connection between matrix computation and pattern counting has deep roots. Early seminal works~\cite{earlywork1, earlywork2, earlywork3} exploit fast matrix multiplication to exactly count specific structured patterns like small cycles and cliques. The \textit{Functional Aggregate Query} (FAQ) framework~\cite{faq} provides a unified theoretical foundation for these concepts, formulating general subgraph counting as a sum-product query:
\begin{equation}
    \label{eq:faq_counting}
    \sum_{x_1,\dots,x_k \in V_G}\prod_{(i,j)\in E_Q} E(x_i, x_j),
\end{equation}
where each variable $x_i$ denotes the mapping of a query vertex to a data vertex, and each factor $E(x_i, x_j)$ enforces edge consistency. Under this framework, matrix multiplication can be seen as a special case of FAQ evaluation, which establishes a natural connection between subgraph counting and matrix-based computation. To solve such queries generally, FAQ proposes \textsf{InsideOut}, an exact algorithm based on variable elimination. Crucially, the FAQ perspective also reveals the fundamental bottleneck shared by all these exact methods: their efficiency is strictly bounded by $O(N^{\mathit{faqw}})$, where $N$ is the data size and $\mathit{faqw}$ is the fractional FAQ-width. While early matrix approaches and \textsf{InsideOut} are effective for low-width topologies, $\mathit{faqw}$ is large for arbitrary or randomly generated subgraphs. Without an exploitable structure, exact evaluation becomes computationally impractical in general, motivating scalable alternatives beyond exact evaluation.

\vspace{1mm}
\noindent \textbf{Candidate Tree-based Counting}.
The candidate tree-based framework supports polynomial-time subgraph homomorphism counting against tree queries~\cite{Fastest, flowsc}.
However, there are two major issues in extending it for subgraph isomorphism counting against generic query graphs.
\textit{First, it ignores the non-tree edges, i.e., edges exist in the query graph but not in its spanning tree}. Ignoring the matching constraints defined by such edges would lead to significant overestimation. 
\textit{Second, it cannot handle the injective mapping constraint required by isomorphism}. Because the candidate tree-based counting only maintains the weight of matching rather than matching states, it cannot enforce injectivity and thus overestimates the counts, as a data vertex might be used multiple times in a matching.

To apply the candidate tree-based counting framework to subgraph isomorphism counting, existing methods develop effective techniques to handle the non-tree edge and injective constraints. \fastest treats the tree homomorphism count as an upper bound and adopts sampling to approximate both the non-tree and injectivity constraints. \flowsc simulates the data flow pattern in the candidate tree counting process and learns the relationship between graph features and subgraph counts, leveraging query features and supervision to implicitly enforce the constraints. However, these methods often suffer from sampling failures in excessive sampling spaces or from insufficient supervision.

\vspace{1mm}
\noindent \textbf{Our Approach and Contributions}.
In this paper, we propose a candidate tree-based counting method, called \underline{A}lgebraic \underline{S}ubgraph \underline{C}ounting (\asc). Since the candidate tree framework inherently satisfies tree-edge constraints, we decouple the remaining requirements for subgraph isomorphism into \textit{non-tree edge constraints} and \textit{isomorphism (injectivity mapping) constraints}.

To handle the non-tree edges, we divide them into two categories: those connecting vertices with the same parent (Type I) and those connecting vertices with different parents (Type II). We first show that it takes exponential time to compute the exact homomorphism count if the Type-I non-tree edges are preserved in the candidate tree framework. We propose a matrix-based method to incorporate Type-I non-tree edge constraints into candidate tree-based counting, delivering high-quality approximations in polynomial time.

For the Type-II non-tree edge constraint, we have a key observation that the compatibility of the candidates of the two endpoints of a Type-II non-tree edge with a valid matching depends on their lowest common ancestor (LCA) in the tree. We incorporate the Type-II non-tree edge constraints into the candidate tree-based counting process by computing the LCA ratios, which are used to discard the invalid weights accumulated at LCAs in the counting phase. Together with the Type-I non-tree edge constraint computation, this leads to an algebraic method for subgraph homomorphism counting. Based on the resulting homomorphism count, we further apply local sampling to enforce injective mapping constraints, producing the final subgraph isomorphism count.

We conduct extensive experiments on $10$ real-world datasets, including two billion-edge graphs. The results demonstrate that \asc can achieve consistently better and more stable performance over the existing solutions under the majority of query settings. 

\section{Preliminary}
\label{sec:Prelim}
\subsection{Problem Definition}

In this paper, we focus on a simple vertex-labeled, undirected, and connected graph $g = (V, E, L, \Sigma)$. Here, $V$ is a set of vertices of $g$, $E \subseteq V \times V$ is a set of edges of $g$, $L$ is a labeling function that assigns a label $L(v) \in \Sigma$ to each vertex $v$ in $g$, and $\Sigma$ is a set of labels.
We use $N(v) = \{v' \in V \mid (v, v') \in E\}$ to denote the neighbors of a vertex $v$, and $d(v) = |N(v)|$ to denote the degree of $v$. In this paper, we use $Q = (V_Q, E_Q, L, \Sigma)$ and $G = (V_G, E_G, L, \Sigma)$ to denote the query graph and the data graph, respectively. Besides, we use $u$ to represent a vertex in $Q$ and $v$ to represent a vertex in $G$.

\begin{definition}[\textbf{Subgraph Isomorphism and Homomorphism}]
\label{def:SubgraphIsomorphism}
    Given a query graph $Q$ and a data graph $G$, subgraph isomorphism is an injective mapping $f: V_Q \mapsto V_G$, such that $\forall u \in V_Q$, $L(u) = L(f(u))$ and $\forall (u, u') \in E_Q$, $(f(u), f(u')) \in E_G$. When injectivity is not required, the mapping $f$ defines a subgraph homomorphism.
\end{definition}

In this paper, we focus on the subgraph isomorphism.
An injective mapping from $Q$ to $G$ is called a \textit{subgraph isomorphic embedding}, also referred to as an \textit{embedding} or a \textit{match} when context is clear.

\begin{definition}[\textbf{Subgraph Matching and Counting}]
    Given a query graph $Q$ and a data graph $G$, subgraph matching is to find all embeddings of $Q$ in $G$, while subgraph counting is to compute the number of embeddings.
\end{definition}

Given the \#P-hardness nature of exact subgraph counting, we focus on approximate solutions for subgraph counting in this paper.

\vspace{1mm}
\noindent \textbf{Problem Statement}. Given a query graph $Q$ and a data graph $G$, we aim to approximate the number of all embeddings of $Q$ in $G$.

\subsection{Candidate Tree-based Counting}

\begin{figure}[t]
    \centering
    \includegraphics[width=\linewidth]{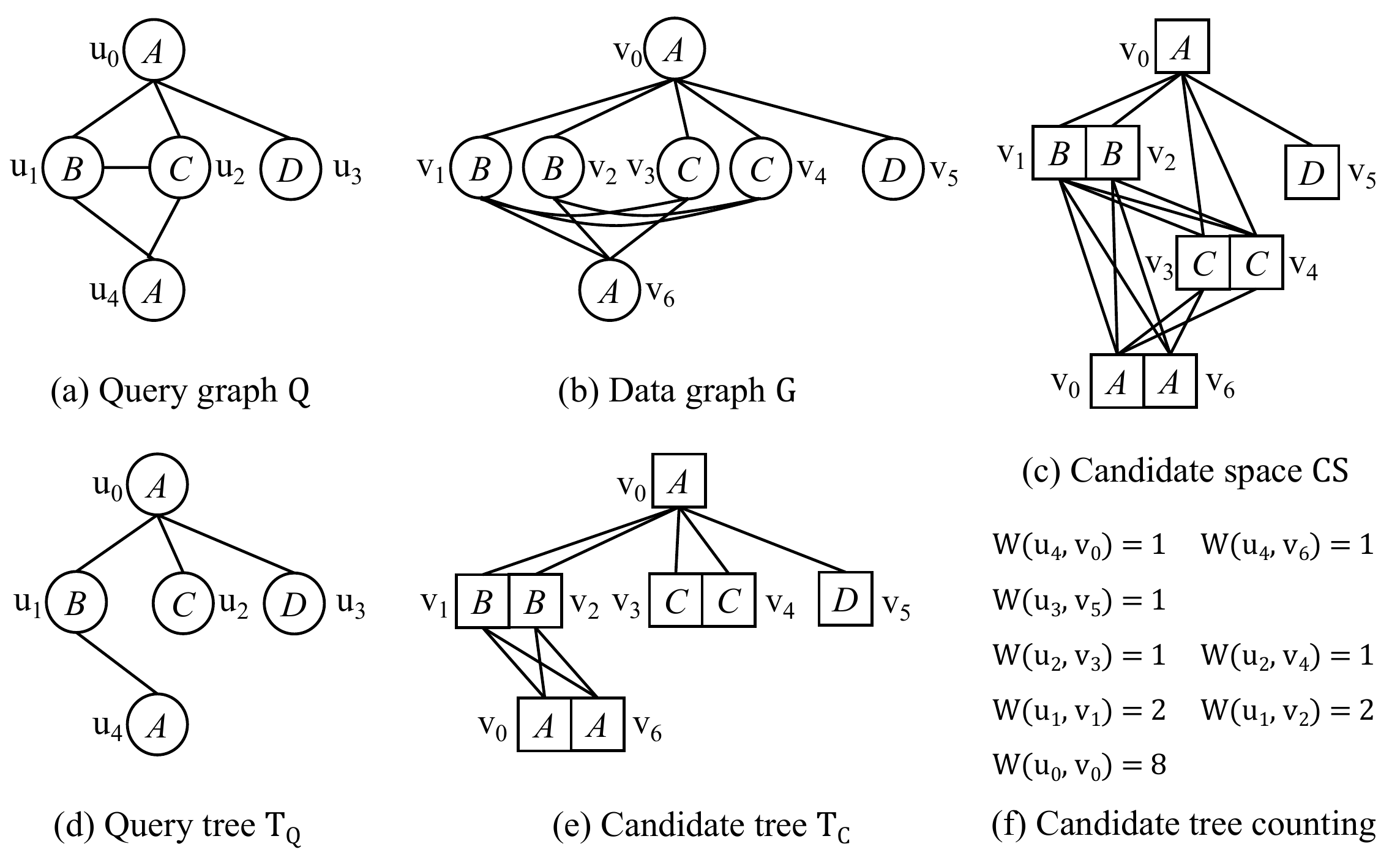}
    \caption{Running example of candidate tree-based counting.}
    \label{fig:non_tree}
\end{figure}

\begin{definition}[\textbf{Candidate Space}]
    Given a query graph $Q$ and a data graph $G$, a candidate space $CS$ maintains a candidate vertex (Resp. edge) set for each vertex (Resp. edge) in $Q$, such that all matches of $Q$ in $G$ are preserved in the solution space defined by $CS$.
    
    $\bullet$ Given a vertex $u\in V_Q$, its candidate vertex set $C(u) \subseteq V_G$ satisfies that, for any embedding $f$, $f(u) \in C(u)$ must hold, but a vertex in $C(u)$ is not necessarily mapped to $u$ in any valid embedding.
    
    $\bullet$ Given an edge $(u, u') \in E_Q$, its candidate edge set $C(u, u') \subseteq E_G$ satisfies that, for any embedding $f$, $(f(u), f(u')) \in C(u, u')$ must hold, but an edge in $C(u, u')$ is not necessarily mapped to $(u, u')$ in any valid embedding.
        
    $\bullet$ Given an edge $(u, u') \in E_Q$ and a candidate vertex $v \in C(u)$, the conditional candidate vertex set $C(u' \mid u, v) \subseteq C(u')$ satisfies that, for any embedding $f$ with $f(u) = v$, $f(u') \in C(u' \mid u, v)$ must hold, but a vertex in $C(u' \mid u, v)$ is not necessarily mapped to $u'$ in any valid embedding.
\end{definition}

\begin{example}
    Consider the query graph $Q$ and data graph $G$ shown in \autoref{fig:non_tree}(a) and (b). We construct the corresponding candidate space $CS$ shown in \autoref{fig:non_tree}(c). In this candidate space, $v_1$ and $v_2$ belong to $C(u_1)$, $(v_0, v_1) \in C(u_0, u_1)$, and $v_1 \in C(u_1 \mid u_2, v_3)$ as $u_1$ can match $v_1$ under the condition $u_2 \mapsto v_3$. 
    
\end{example}

To improve the efficiency and accuracy of subgraph counting, existing methods often build a compact candidate space by removing the invalid vertices and edges through candidate filtering algorithms ~\cite{GraphQL, flowsc, Fastest, CFLMatch}. Then, the subsequent counting process is performed against a more compact candidate space instead of the original data graph, without sacrificing completeness. The candidate tree structure organizes the candidate space hierarchically to assist subgraph counting.

\begin{definition}[\textbf{Candidate Tree}]
    Let $T_Q$ be a spanning tree of the query graph $Q$. The candidate tree $T_C$ is obtained by organizing the candidate space according to the structure of $T_Q$: each query vertex $u$ in $T_Q$ is associated with its candidate set $C(u)$, and a tree edge $(u, u')$ in $T_Q$ is replaced by the candidate edges between $C(u)$ and $C(u')$. Consequently, $T_C$ encodes all tree homomorphisms of $T_Q$ in the candidate space. 
\end{definition}

\begin{example}
\label{ex:candidate_trees}
    For the query graph $Q$ in \autoref{fig:non_tree}(a), we construct its BFS spanning tree $T_Q$ as shown in \autoref{fig:non_tree}(d). Its corresponding candidate tree $T_C$ is illustrated in \autoref{fig:non_tree}(e). In particular, $(\langle u_0, v_0 \rangle, \langle u_1, v_1 \rangle, \allowbreak \langle u_2, v_3 \rangle, \langle u_3, v_5 \rangle, \langle u_4, v_0\rangle)$ forms a homomorphism embedding of $T_Q$.
\end{example}

\noindent \textbf{General Idea of Candidate Tree-based Counting}.
The candidate tree-based counting framework utilizes dynamic programming to compute the number of homomorphisms of $T_Q$ in polynomial time, providing a tight upper bound for subgraph isomorphism counting of $Q$.
In this framework, homomorphism counts are accumulated in a bottom-up manner along the tree edges of $T_C$.
During the processing, each matching pair maintains a weight as the number of homomorphism matches of its subtrees.
At the root level, the weights of all root matching pairs are collected to obtain the homomorphism count for the $T_Q$. 
More specifically, for a vertex $u \in V_Q$, we use $N_c(u)$ to denote its children in $T_Q$.
For each $v \in C(u)$, we use a weight $W(u, v)$ to represent the number of homomorphisms of the subtree $T_Q(u)$ (rooted at $u$) contained in the candidate subtree $T_C(v)$ (rooted at $v$). Then, $W(u, v)$ is updated as follows:
\begin{equation}
    \label{eq:tree_counting}
    W(u, v) = \prod\limits_{u_c \in N_c(u)} \sum\limits_{v_c \in C(u_c \mid u, v)} W(u_c, v_c).
\end{equation}
\noindent Each matching pair updates its weight with the weights of its child matching pairs, and $W(u, v) = 1$ if $u$ is a leaf vertex. The counting process proceeds from the leaves to the root, with an overall $O(|E_Q||E_C|)$ time~\cite{Fastest}.

\begin{example}
\label{ex:counting}
    Continuing the \autoref{ex:candidate_trees}, we compute the candidate tree count in \autoref{fig:non_tree}(e) and the updated weights for each matching pair are shown in \autoref{fig:non_tree}(f). Candidate tree counting proceeds in a bottom-up manner, starting from the leaves. All leaf matching pairs have a weight of $1$, e.g., $W(u_4, v_0) = 1$, and $W(u_4, v_6) = 1$. Following the counting rule in \autoref{eq:tree_counting}, $W(u_1, v_1) = (W(u_4, v_0) + W(u_4, v_6))$, so we have $W(u_1, v_1) = 2$. Similarly, $W(u_1, v_2) = 2$. At the root level, $W(u_0, v_0) = (W(u_1, v_1) + W(u_1, v_2)) \times (W(u_2, v_3) + W(u_2, v_4)) \times W(u_3, v_5) = 8$, so the tree homomorphism count is $8$.
\end{example}

\vspace{1mm}
\noindent \textbf{Issues for Subgraph Isomorphism}.
The candidate tree-based counting method offers an efficient paradigm for subgraph homomorphism counting against tree queries. 
However, there are two major issues in extending it for subgraph isomorphism counting against generic query graphs.

\vspace{1mm}
$\bullet$ \textit{Issue 1: Non-tree Edge Constraint}.
\label{issue:non-tree}
As shown in \autoref{eq:tree_counting}, for each matching pair $(u,v)$, we compute, for each child vertex $u_c \in N_c(u)$, the sum of $W(u_c, v_c)$ over all $v_c \in C(u_c \mid u, v)$, and then take the product of these per-child sums to obtain $W(u, v)$.
However, in the processing, the non-tree edges are ignored, which may significantly affect the accumulated weight for the parent vertex, since every edge in the query graph must be matched. 
Here, a non-tree edge refers to an edge existing in the query graph but not in the query tree.

\begin{example}\label{ex:non_tree_constraints}
    Comparing the $Q$ in \autoref{fig:non_tree}(a) with its BFS tree $T_Q$ in \autoref{fig:non_tree}(d), there are $2$ non-tree edges $(u_1, u_2)$ and $(u_2, u_4)$, that are not included in $T_Q$. As in \autoref{ex:counting},
    the resulting candidate-tree count is $8$. However, when graph homomorphism constraints are considered, the true count is $4$. The discrepancy is caused by non-tree constraints.
    Considering the constraints imposed by the non-tree edge $(u_1, u_2)$, the mappings $u_1 \mapsto v_2$ and $u_2 \mapsto v_3$ are conflicting. Nevertheless, the weight update of $W(u_0, v_0)$ implicitly assumes that every candidate in $C(u_1 \mid u_0, v_0)$ connects to every candidate in $C(u_2 \mid u_0, v_0)$, which does not hold in practice. Similarly, the constraint imposed by $(u_2, u_4)$ further restricts valid combinations and should also limit the weights propagation in the candidate tree counting.
\end{example}

$\bullet$ \textit{Issue 2: Isomorphism Constraint}.
In the processing of candidate tree-based counting, only the weights of matching pairs are maintained, rather than the explicit matching states.
This implies that a data vertex may correspond to multiple query vertices, resulting in a non-injective mapping, which overestimates the number of subgraph isomorphism embeddings. We use the term \textit{isomorphism (iso) constraint} to refer to the injectivity mapping constraint.

\begin{example}\label{ex:iso_constraints}
    Continuing the \autoref{ex:non_tree_constraints}, when vertex labels and vertex degrees are considered as filtering criteria~\cite{Ullmann1976, CFLMatch}, $v_0$ is contained in both $C(u_0)$ and $C(u_4)$ as shown in \autoref{fig:non_tree}(c). Consequently, in the current candidate tree, the weight of leaf $v_0$ in $W(u_4, v_0)$ is propagated upward to the root $v_0$ in $W(u_0, v_0)$, causing $v_0$ to be counted twice and resulting in an overestimation in subgraph isomorphism count.
\end{example}

\vspace{1mm}
\noindent \textbf{Remedy of Existing Solutions}.
To support generic subgraph counting, prior works extend candidate tree counting by approximating non-tree and injective constraints. \fastest~\cite{Fastest} estimates the final count by sampling candidate trees and evaluating whether they form valid query matches. However, its performance deteriorates rapidly when the sampling space becomes prohibitively large, which often happens on larger datasets and leads to sampling failures. Alternatively, \flowsc~\cite{flowsc} employs supervised learning with cross-graph attention to implicitly emulate matching checks. However, its effectiveness relies on sufficient ground truth, which is costly to obtain for hard query workloads.

\section{Our Approach}
\label{sec:overview}
\noindent \textbf{Overview}.
Built upon the candidate tree framework, \asc first incorporates Type-I constraints via a matrix-based scheme (\autoref{sec:Module1}), and utilizes an LCA ratio mechanism to filter out invalid combinatorial weights induced by Type-II constraints (\autoref{sec:lca}). Based on the resulting algebraic homomorphism count, we finally enforce the injectivity constraints via local sampling on subgraphs induced by same-label vertices (\autoref{sec:iso}).

\noindent \textbf{Algorithm Details}.
Algorithm~\ref{alg:ASC} summarizes the counting process of \asc. Given a query graph $Q$ and a data graph $G$, we first construct and refine the candidate space $CS$ using the filtering method \bp~\cite{flowsc}  (\autoref{line:filtering}), and build a BFS spanning tree $T_Q$ of $Q$. 
We choose BFS rather than DFS spanning trees because the former generates more Type-I non-tree edges, whereas the latter generates only Type-II non-tree edges.
Type-I non-tree edges are easier to handle than Type-II ones, as discussed in the following two sections.
Then we compute the LCA ratio for each LCA candidate (\autoref{line:lca_ratio}).

In the counting phase, the parent weight $W(u, v)$ is updated from its child vertices $N_c(u)$ using a matrix computation-based counting method \textsf{SequenceCounting} that incorporates Type-I non-tree constraints in $Q[N_c(u)]$ (\autoref{line:constrained_counting}), where $Q[N_c(u)]$ denotes the subgraph of $Q$ induced by $N_c(u)$.
The Type-II non-tree edge constraints are handled by applying the LCA ratios at LCA vertices (\autoref{line:apply_lca_ratio}). We update weights in reversed BFS order and aggregate them at the root (\autoref{line:homo_counts}) to obtain the subgraph homomorphism count\footnote{As \bp applies injectivity filtering, this is not a strict homomorphism count.}. An isomorphism ratio is then estimated using a local sampling method (\autoref{line:iso_ratio}), and the final subgraph isomorphism count is obtained by multiplying the homomorphism count by this ratio (\autoref{line:iso_counts}).

\begin{algorithm}[t]
    \caption{\textsc{Algebraic Subgraph Counting}}
    \rm\sffamily
    \label{alg:ASC}
    \KwIn{A query graph $Q$, a data graph $G$}
    \KwOut{Estimated subgraph count of $Q$ in $G$}
    \SetKwProg{Fn}{Function}{:}{}

    CS $\gets$ \bp($Q$, $G$)\;
    \label{line:filtering}
    $T_Q$ $\gets$ BuildTree($Q$)\;
    LCARatio($CS$, $T_Q$)\;
    \label{line:lca_ratio}

    \ForEach{$u$ in reversed BFS Traversal of $T_Q$}{
        $s \gets$ BuildSequence($Q[N_c(u)]$)\;
        \ForEach{$v \in C(u)$}{
            $W(u, v) \gets $ SequenceCounting($u,v$,$CS$,$s$)\;
            \label{line:constrained_counting}
            \If{$u$ is a LCA vertex} {
                $W(u, v) \gets W(u, v) \times LCA(u, v)$\;
                \label{line:apply_lca_ratio}
            }
        }
    }

    HomoCounts $\gets$ 0\;
    \ForEach{$v_r \in C(u_r)$ where $u_r$ is the root of $T_Q$}{
        HomoCounts $+= W(u_r, v_r)$\;  
        \label{line:homo_counts}
    }

    IsoRatio $\gets$ LocalSampling($CS$, $Q$)\;
    \label{line:iso_ratio}

    IsoCounts $\gets$ HomoCounts $\times$ IsoRatio\;
    \label{line:iso_counts}
    \Return{IsoCounts}
\end{algorithm}

\section{Handling Type-I Non-tree Edge Constraints}
\label{sec:Module1}
In this section, we introduce the key component of \asc, which is a matrix-based counting model that can approximate the Type-I non-tree edge constraint in candidate tree-based counting.

The tree counting formulation in \autoref{eq:tree_counting} implicitly assumes child vertices are mutually independent, leading to error accumulation when inter-dependencies (i.e., non-tree edge constraints) exist. While exact counting for these constraints among child vertices typically incurs exponential complexity, we develop an algebraic method to simulate this process in polynomial time. By transforming candidate connectivity into transition matrices and candidate weights into weight vectors, we propagate weights via matrix multiplications to simulate weight accumulation in the Cartesian product space (\autoref{sec:simple_case}). For complex structures like cycles among the children, we utilize weight matrices to encode both single and multiple indirect constraints (\autoref{sec:generic_case}). By traversing the children through carefully designed matrix operations, our method incorporates Type-I non-tree constraints into the tree counting, with the resulting vector sum yielding the weight for a parent vertex.

\subsection{A Simplified Counting Case}
\label{sec:simple_case}

\vspace{1mm}
\noindent \textbf{Fixing Candidate Tree-based Counting}.
\autoref{eq:tree_counting} formulates a tree-based counting paradigm where weight updates are computed by decoupling each child branch, implicitly assuming that child vertices are mutually isolated.
To accommodate it to the case where the child vertices are connected (i.e., non-tree edges), we need to inspect each matching combination for the child vertices. 
Specifically, given a query vertex $u$ and its matching vertex $v$,
let $\mathbf{\Omega} = \times_{u_c \in N_c(u)} C(u_c \mid u, v)$ be the Cartesian product of the conditional candidates of $u$'s child vertices. Clearly, $\mathbf{\Omega}$ defines a complete matching space for the child vertices of $u$ when $u \mapsto v$. Based on $\mathbf{\Omega}$, we can compute the weight as follows:
\begin{equation}
    \label{eq:counting_unit}
    W(u, v) = \sum_{\mathbf{v} \in \mathbf{\Omega}} \left( \sigma(\mathbf{v}) \cdot \prod_{u_c \in N_c(u)} W(u_c, \mathbf{v}[u_c]) \right).
\end{equation}
Here, each combination $\mathbf{v} \in \mathbf{\Omega}$ represents a possible matching for vertices in $N_c(u)$, and $\mathbf{v}[u_c]$ denotes the candidate vertex mapped to $u_c$. The term $\sigma$ is an indicator function with value $1$ if the subgraph induced by $\mathbf{v}$ is a homomorphism of $Q[N_c(u)]$, and $0$ otherwise.

\autoref{eq:counting_unit} follows the same sum-product structure as the FAQ based subgraph counting (\autoref{eq:faq_counting}), where the summation enumerates possible vertex mappings and the product enforces structural consistency. Compared with the standard FAQ expression, \autoref{eq:counting_unit} further incorporates candidate-tree-based weights, transforming the generic sum-product query into a recursive weighted aggregation process.

Direct evaluation of \autoref{eq:counting_unit} requires enumerating all combinations in the matching space $\mathbf{\Omega}$, resulting in $O(n^{|N_c(u)|})$ time, which is exponential in $|N_c(u)|$, assuming that the maximum size of $C(u_c \mid u, v)$ is $n$. In the FAQ framework, the \textsf{InsideOut} algorithm reduces the cost through optimized variable elimination, but its complexity still depends exponentially on the fractional FAQ-width of the subquery. To address this challenge, we propose a matrix-based approximation that avoids exhaustive enumeration and produces high-quality counting results in polynomial time.

\vspace{1mm}
\noindent \textbf{Intuition of Matrix Computation}.
The main innovation of our method is transforming the subgraph homomorphism counting into a sequence of matrix multiplication operations. 
Intuitively, adjacency matrix multiplication provides an efficient mechanism for evaluating connectivity between vertex sets. Moreover, its computational rule naturally aligns with Cartesian product operations across vertex sets, where each resulting entry counts the number of connecting paths between the corresponding pair of vertices in the candidate space.
With a carefully designed computation strategy, we can incorporate non-tree constraints via algebraic methods.

\vspace{1mm}
\noindent \textbf{Constrained Sequence}.
Given a query vertex $u$, its child vertices can induce a set of connected components. 
For example, in the $Q$ of \autoref{fig:non_tree}(a), $\{u_1, u_2 \}$ and $\{u_3\}$ are two connected components.
According to the above discussion, each connected component can be processed independently. In the following, we focus on the non-trivial connected components (i.e., containing at least $2$ vertices), since the non-tree edges only exist in such components.

\begin{definition}[\textbf{Constrained Sequence}]
Given a query vertex $u$ in the query tree $T_Q$, let $S$ be a non-trivial connected component induced by a set of vertices in $N_c(u)$ from $Q$.
A constrained sequence $s$ is an ordered arrangement of vertices in $S$ with the following properties.

$\bullet$ Given a vertex $u_c \in S$, a vertex appearing before (Resp., after) $u_c$ in $s$ is called the predecessor (Resp., successor) of $u_c$.

$\bullet$ Given two vertices $u_c', u_c \in S$, $u_c'$ is called the direct predecessor (Resp., direct successor) of $u_c$ if it  appears immediately before (Resp., after) $u_c$ in $s$; Otherwise, $u_c'$ is an indirect predecessor (Resp., indirect successor) of $u_c$ if $u_c'$ is a neighbor of $u_c$.  
\end{definition}

Based on the second property, we note that there must be a non-tree edge between the indirect predecessor-successor vertices, while it is not necessary between the direct ones.

\begin{example}
\label{ex:sequence}
    Consider the $T_{Q_1}$ in \autoref{fig:sequence_counting}(a), where $\langle u_1,u_2,u_3, u_4 \rangle$ is a constrained sequence generated by vertices in $N_c(u_0)$.
    Clearly, $u_1$ is a direct predecessor of $u_2$ and an indirect predecessor of $u_4$.
\end{example}

\vspace{1mm}
\noindent \textbf{Counting via Matrix Multiplication Propagation}.
For ease of exposition, we start by considering a simplified case where $s$ has no indirect predecessor-successor vertices.
Given a query tree $T_Q$ and its corresponding candidate trees $T_C$, let $u$ and $v$ be a pair of matching vertices in the trees.
To compute the weight $W(u,v)$ in \autoref{eq:counting_unit}, we maintain a weight vector for each child vertex $u_c$ of $u$ and a transition matrix for each direct predecessor-successor vertex pair of $u$, which are defined as follows.

\begin{definition}[\textbf{Weight Vector}]
\label{def:weight_vector}
    Given a vertex $u_c$ in the constrained sequence $s$, the weight vector of $u_c$, denoted as $\mathbf{V}$, is a column vector with $|C(u_c \mid u, v)|$ entries, which are initialized as $W(u_c,v_c)$ for each $v_c \in C(u_c \mid u, v)$.
\end{definition}

\begin{definition}[\textbf{Transition Matrix}]
    Given a direct 
    predecessor-successor vertex pair $u_c'$ and $u_c$ in the constrained sequence $s$, the transition matrix between $u_c'$ and $u_c$, denoted as $\mathbf{A}$, is constructed with rows and columns corresponding to $C(u'_c \mid u, v)$ and $C(u_c \mid u, v)$, respectively. 
    This matrix $\mathbf{A}$ encodes the non-tree constraint between $u'_c$ and $u_c$ within the candidate space.
    If a non-tree edge exists between $u'_c$ and $u_c$, then $\mathbf{A}_{ij} = 1$ if there exists a candidate edge between the $i$-th candidate of $u'_c$ and the $j$-th candidate of $u_c$, and $\mathbf{A}_{ij} = 0$ otherwise. If no non-tree edge exists, $\mathbf{A}$ contains all ones, representing that all candidate combinations between $u'_c$ and $u_c$ are considered valid. 
\end{definition}

\begin{example}
Consider the query tree $T_Q$ in \autoref{fig:sequence_counting}(a) and the candidate trees $T_C$ in \autoref{fig:sequence_counting}(b), where $\langle u_1, u_2, u_3, u_4 \rangle$ is a constrained sequence. Note here that we ignore the red edges at this moment.
The weight vectors of vertices are shown in \autoref{fig:sequence_counting}(c). Taking $u_2$ for example, since it is a leaf vertex and has two candidates, i.e., $v_2$ and $v_3$, its weight vector is initialized as $[[1],[1]]$. The transition matrices are shown in \autoref{fig:sequence_counting}(d). 
\end{example}

With the weight vector and transition matrix, we can compute $W(u,v)$ as follows. We start from the last vertex in the constrained sequence and propagate backwards. At each step, the weight vector of the successor vertex is multiplied by the transition matrix and then combined with the predecessor's own weight vector via a Hadamard product.
Therefore, the counts are propagated along the non-tree candidate edges and accumulated into the predecessor's weights. 
We repeat this processing until the head is reached, where 
the sum of the final vector's entries is used to update $W(u, v)$.

\begin{figure*}[t] 
    \centering
    \includegraphics[width=1.0\textwidth]{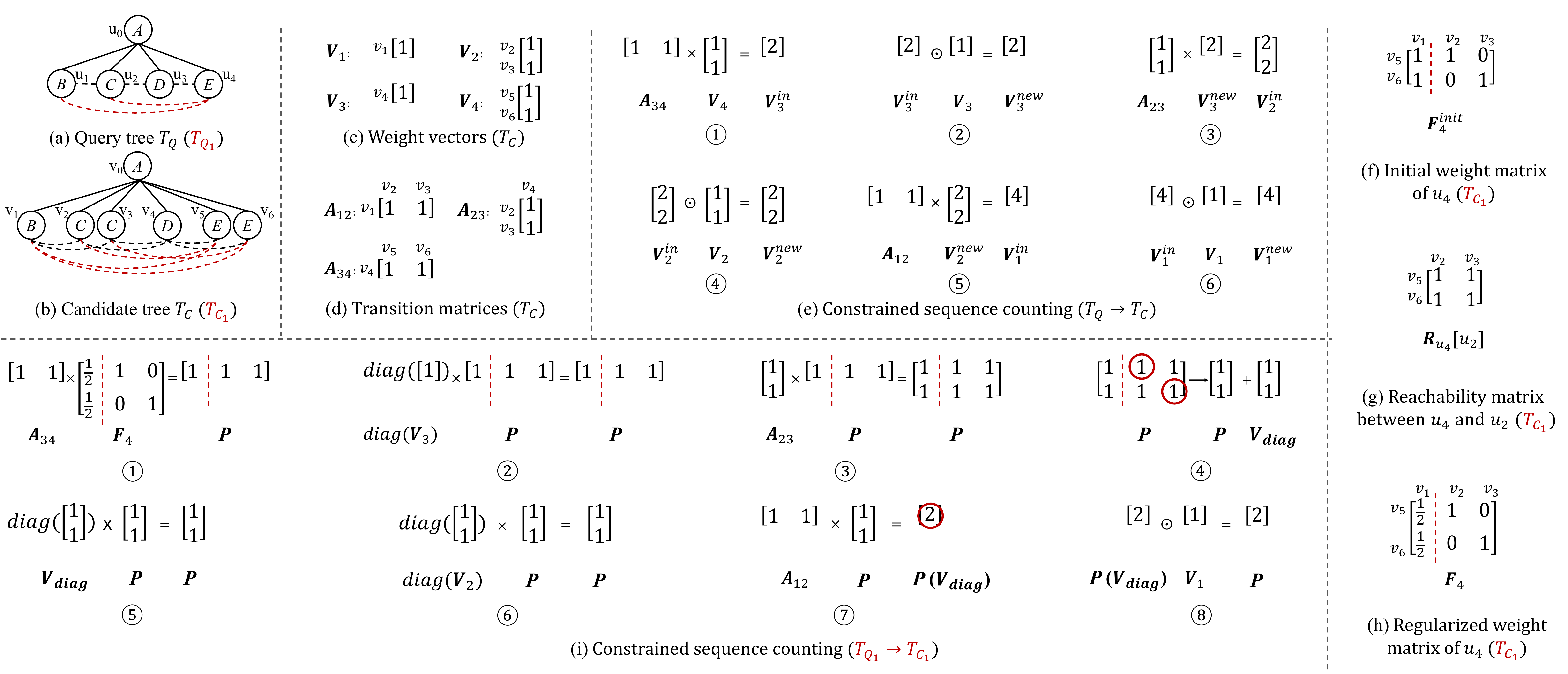}
    \caption{Example of matrix-based counting: dashed edges denote non-tree edges and red edges present only in $T_{Q_1}$ and $T_{C_1}$.}
    \label{fig:sequence_counting}
\end{figure*}

\begin{example} \label{ex:basic_counting}
    The constrained sequence counting starts from the end of sequence, proceeds backward along the sequence, as shown in \autoref{fig:sequence_counting}(e): in Step \textcircled{1}, the weight vector $\mathbf{V}_4$ of $C(u_4 \mid u_0, v_0)$ is multiplied by $\mathbf{A}_{34}$ to produce the input vector $\mathbf{V}_3^{in}$ for $C(u_3 \mid u_0, v_0)$, which are then multiplied element-wise with the original weight vector $\mathbf{V}_3$ in Step \textcircled{2} to obtain the updated vector $\mathbf{V}_3^{new}$. This process continues backwards along the sequence, following the Steps \textcircled{3}-\textcircled{6}. Then, at the head $u_1$, the sum of the entries in $\mathbf{V}_1^{new}$ yields the count of the constrained sequence, which is $4$ and exactly equals the homomorphism count of the constrained sequence.
\end{example}

\begin{lemma}
    Given a constrained sequence $s$ with parent vertex $u$ matched to a candidate $v$, the matrix-based method returns the exact homomorphism count of the induced subgraph $Q[s]$ in the candidate space conditioned on $u \mapsto v$, assuming there are no indirect predecessor-successor connections and no false positive candidates.
\end{lemma}

\begin{proof}
    For any predecessor-successor pair $(u'_c, u_c)$ in the sequence $s$, the matrix-vector product $\mathbf{V}_{u'_c}^{in} = \mathbf{A}_{u'_cu_c}\times\mathbf{V}_{u_c}$ exclusively collects weights from candidates in $C(u_c \mid u, v)$ that form valid matching with candidates in $C(u'_c \mid u, v)$ under the non-tree constraint of $(u'_c, u_c)$. 
    The Hadamard product $\mathbf{V}_{u'_c}^{new} = \mathbf{V}_{u'_c} \odot \mathbf{V}_{u'_c}^{in}$ updates the weight vector of $u'_c$ by multiplying each candidate's weight in $\mathbf{V}_{u'_c}$ 
    with the sum of weights propagated from the candidates in $C(u_c \mid u, v)$ via $\mathbf{A}_{u'_cu_c}$.
    Since $C(u'_c \mid u, v)$ and $C(u_c \mid u, v)$ are restricted to candidates under the condition that $u$ matches $v$, this weight accumulation simultaneously satisfies both constraints imposed by the tree-edges from $u$ and the non-tree edge of $(u'_c, u_c)$.
    
    \noindent\textbf{Induction:} By repeating the propagation step backwards through the sequence, each step incorporates exactly one non-tree edge constraint. The final value is a sum of product terms where each term represents a combination of candidates that satisfies all tree and non-tree constraints. Since each valid combination is counted exactly once through these algebraic operations, the final summation of entries in the head vector yields the exact total count.
\end{proof}

\subsection{Extending to Generic Case}
\label{sec:generic_case}

\noindent \textbf{Motivation}.
When a query vertex has an indirect predecessor, the indirect and direct edges along the sequence jointly impose constraints on the candidate matching. For example, in the query tree $T_{Q_1}$ shown in \autoref{fig:sequence_counting}(a), where the red edges are included in the tree, $u_4$ has $2$ indirect predecessors $u_1$ and $u_2$. The indirect edges $(u_1, u_4), (u_2, u_4)$ and the sequence $\langle u_1, u_2, u_3, u_4 \rangle$ jointly impose matching constraints on this sequence.

We propose to use a weight matrix to encode the indirect constraints on a vertex, and further approximate the interactions among these constraints through an enhanced regularized weight matrix. With carefully designed propagation rules, the weight matrix accumulates valid weights during the traversal of the constrained sequence to support approximate counting with indirect constraints.

\vspace{1mm}
\noindent \textbf{Incorporating Indirect Predecessor Constraints}.
In the following, we proceed to deal with the indirect predecessor-successor relationships between vertices presented in the constrained sequence. For a vertex with indirect predecessors, we need to properly take the connection constraints brought by such predecessors into consideration during the matrix propagation. To this end, we introduce the concept of \textit{weight matrix}, which is essentially an advanced version of the weight vector by integrating both the candidate weights and the connection status with the indirect predecessors. 

\begin{definition}[\textbf{Weight Matrix}]
\label{def:weight_matrix}
    Given a vertex $u_c$ in the constrained sequence $s$ with $m$ indirect predecessors $\langle u_{p_1}, ... ,u_{p_m} \rangle$ ordered as in $s$, the weight matrix of $u_c$, $\mathbf{F}$, consists of $m$ column-wise concatenated blocks, where the $k$-th block $\mathbf{F}[u_{p_k}]$ corresponds to $u_{p_k}$:

$\bullet$ For an indirect predecessor $u_{p}$, the corresponding block $\mathbf{F}[u_p]$ is a matrix with $|C(u_c \mid u, v)|$ rows and $|C(u_{p} \mid u, v)|$ columns, where an entry $\mathbf{F}[u_p]_{ij}=1$ if there is a candidate edge between the $i$-th candidate of $u_c$ and the $j$-th candidate of $u_{p}$, and $0$ otherwise.

$\bullet$ For the block corresponding to the first indirect predecessor $u_{p_1}$, we further set its nonzero entries to the corresponding candidate weights $W(u_c, v_c)$ of $u_c$. The weight accumulation in the propagation process is performed through this primary block, while the remaining blocks are used to verify connectivity for non-tree constraints. This ensures that the candidate weights of $u_c$ are accumulated exactly once in the Cartesian-product space. 
\end{definition}

\begin{example}
\label{ex:initial_feature_matrix}
    For the constrained sequence $\langle u_1, u_2, u_3, u_4 \rangle$ in $T_{Q_1}$ shown in \autoref{fig:sequence_counting}(a) and its corresponding candidate trees $T_{C_1}$ in \autoref{fig:sequence_counting}(b), the initial weight matrix $\mathbf{F}_4^{init}$ of $u_4$ is presented in \autoref{fig:sequence_counting}(f), where the red dashed line separates the blocks corresponding to different indirect predecessors. 
    Since $u_4$ has two indirect predecessors $u_1$ and $u_2$, its weight matrix has two blocks.
    Considering $u_1$, since there is only one vertex $v_1$ in $C(u_1 \mid u_0, v_0)$ which connects to both candidates $v_5$ and $v_6$ in $C(u_4 \mid u_0, v_0)$, the corresponding block is a $2 \times 1$ matrix of ones, representing both connectivity and candidate weights.
    The block for $u_2$ is constructed similarly based on connectivity only.

\end{example}

\vspace{1mm}
\noindent \textbf{Regularizing Weight Matrix}.
Recall that the overall counting process is formulated as a matrix-based propagation backwards along the constrained sequence. 
Based on \autoref{def:weight_matrix}, the candidate weights of the successor and the non-tree constraint from its first indirect predecessor are embedded in the first block of the successor's weight matrix $\mathbf{F}$, where weight accumulation in propagation takes place.
The constraints from other indirect predecessors are encoded in separate blocks. 
These constraints from different indirect predecessors cannot be applied to the successor's candidates simultaneously for valid weight accumulation during the propagation process.

\begin{example}
    In $T_{Q_1}$ shown in \autoref{fig:sequence_counting}(a), both $u_1$ and $u_2$ are indirect predecessors of $u_4$. In the candidate tree $T_{C_1}$ in \autoref{fig:sequence_counting}(b), the weights of $u_4$'s candidates $v_5$ and $v_6$ can propagate backwards along the sequence to $v_2$ and $v_3$. This propagation satisfies the indirect constraint between $u_1$ and $u_4$, since $v_1$ connects to both $v_5$ and $v_6$, but it violates the indirect constraint between $u_2$ and $u_4$, because $v_2$ is not adjacent to $v_6$, and $v_3$ is not adjacent to $v_5$. The indirect constraints of $(u_1, u_4)$ and $(u_2, u_4)$ are encoded in different blocks of the weight matrix, which may not be satisfied simultaneously on $v_5$ and $v_6$.
\end{example}

To jointly enforce the constraints from multiple indirect predecessors on a successor without introducing complex multi-constraint checks during propagation, we develop regularization techniques for the weight matrix.
Specifically, we apply a scaling ratio matrix to the block of the first indirect predecessor in the weight matrix to approximately offset invalid weights before the propagation begins. Essentially, each ratio estimates how much of the weight propagated from a successor candidate to the corresponding candidate of the first indirect predecessor remains valid after jointly enforcing all indirect-predecessor constraints along the sequence.
To facilitate the computation, we make use of the following definition \textit{reachability matrix}, which records the reachability situation between the candidates of two query vertices. 

\begin{definition}[\textbf{Reachability Matrix}]
Given two vertices $u_c$ and $u_c'$ in the constrained sequence,
the reachability matrix between $u_c$ and $u_c'$, denoted as $\mathbf{R}_{u_c}[u_c']$, maintains the reachability from vertices in $C(u_c \mid u, v)$ to $C(u_c' \mid u, v)$ along the constrained sequence $s$, which can be easily obtained via successive multiplying the transition matrices between consecutive vertices in the constrained sequence. 
\end{definition}

\begin{example}
    \autoref{fig:sequence_counting}(g) shows the reachability matrix $\mathbf{R}_{u_4}[u_2]$ from candidate of $u_4$ to the candidates of $u_2$, which is obtained by binarizing $(\mathbf{A}_{23} \times \mathbf{A}_{34})^{\top}$, where all non-zero entries are set to $1$. It indicates that both candidates $v_5$ and $v_6$ can reach $v_2$ and $v_3$ in the candidate space along the constrained sequence.
\end{example}

Now, given a vertex $u_c$ in the constrained sequence with indirect predecessors $\langle u_{p_1}, ... ,u_{p_m} \rangle$ and its initial weight matrix $\mathbf{F}$, we use $\mathbf{F}[u_{p_i}]$ to denote the matrix block corresponding to the indirect predecessor $u_{p_i}$.
With the reachability matrix, we update the weight matrix as follows.
 
\begin{equation}
\label{eq:updated_feature}
\begin{aligned}
    \mathbf{F}[u_{p_i}] \leftarrow & \mathbf{F}[u_{p_i}] \odot ( \text{diag}(\mathbf{R}_{u_c}[u_{p_{i+1}}] \cdot \mathbf{1})^{-1} \cdot \\
    & ( (\mathbf{F}[u_{p_{i+1}}] \odot \mathbf{R}_{u_c}[u_{p_{i+1}}] ) \cdot \mathbf{R}_{u_{p_{i+1}}}[u_{p_i}] ) ).
\end{aligned}
\end{equation}

In \autoref{eq:updated_feature}, the state at $u_{p_i}$ depends on $u_{p_{i+1}}$. We examine the reachability of each candidate of $u_c$ to candidates at $u_{p_{i+1}}$ with $\mathbf{R}_{u_c}[u_{p_{i+1}}]$. Reachable candidates at $u_{p_{i+1}}$ receive the accumulated weights propagated from $u_c$'s candidates. Accordingly, multiplying by an all-ones vector $\mathbf{1}$ yields the number of reachable candidates at $u_{p_{i+1}}$ from each candidate of $u_c$. The resulting vector is broadcast via the $diag(\cdot)$ operation to form a diagonal matrix, whose diagonal entries serve as the denominator for the corresponding row in the ratio computation. The numerator accounts for how much of the accumulated weight propagated from $u_c$'s candidates satisfies the indirect constraints at $u_{p_{i+1}}$ and can further propagate to $u_{p_i}$. This is computed in the second line of \autoref{eq:updated_feature}: $\mathbf{F}[u_{p_{i+1}}]$ encodes the indirect constraints of $u_{p_{i+1}}$, while $\mathbf{R}_{u_c}[u_{p_{i+1}}]$ captures reachability from $u_c$ to $u_{p_{i+1}}$. Candidates that are both reachable and indirect constraint-satisfying are propagated to $u_{p_i}$ through $\mathbf{R}_{u_{p_{i+1}}}[u_{p_i}]$, forming the numerator of the ratio computation. The resulting ratio matrix is combined with $\mathbf{F}[u_{p_i}]$ via the Hadamard product, thereby incorporating the indirect constraints at $u_{p_i}$ and producing an updated matrix block, which is then used to update $u_{p_{i-1}}$. Since only the first block of $\mathbf{F}$ contains candidate weights of $u_c$, we retain only the regularized update of this block in \autoref{eq:updated_feature}. The remaining blocks encode connectivity, where only the zero/non-zero distinction matters, so their original values are kept.

\begin{example}
\label{ex:updated_initial_feature}
    Based on the initial feature matrix in \autoref{ex:initial_feature_matrix},
    we apply a scaling ratio matrix on the block $\mathbf{F}_4^{init}[u_1]$ of $u_1$.
    For the reachability matrix $\mathbf{R}_{u_4}[u_2]$ shown in \autoref{fig:sequence_counting}(g), we compute the row sums, resulting in the vector $[[2], [2]]$, which serves as the per-row denominators of the ratio matrix. Next, we take the Hadamard product of $\mathbf{R}_{u_4}[u_2]$ and $\mathbf{F}_4^{init}[u_2]$, obtaining $[[1, 0], [0, 1]]$, which incorporates the indirect constraint on $(u_2, u_4)$. To evaluate how much of the accumulated weight can further propagate to $C(u_1 \mid u_0, v_0)$, we multiply the obtained matrix by $\mathbf{R}_{u_2}[u_1] = [[1], [1]]$, yielding $[[1], [1]]$. Dividing each entry by its corresponding row denominator produces the ratio matrix $[[1/2], [1/2]]$. Finally, the ratio matrix is applied to $\mathbf{F}_4^{init}[u_1]$ via Hadamard product, which further incorporates the indirect constraint between $(u_1, u_4)$. The resulting matrix is the final feature matrix $\mathbf{F}_4$, as shown in \autoref{fig:sequence_counting}(h).
\end{example}

\vspace{1mm}
\noindent \textbf{Propagating Weight Matrix}.
After setting the weight matrix properly, we propagate it via matrix multiplication along the reverse direction of the constrained sequence until reaching the first vertex, where the accumulated weights are extracted as the homomorphism count.
Because each vertex in the constrained sequence might have its own weight matrix, to avoid ambiguity, we call the propagated weight matrix as \textit{passing matrix} $\mathbf{P}$.
During the propagation, we need to execute the following two types of operations accordingly, including \textit{weight aggregation} and \textit{weight extraction}.

\vspace{1mm}
$\bullet$ \textit{Weight Aggregation.}
Let $u_c$ be the current vertex in the constrained sequence, and $u_c'$ its direct predecessor during the propagation. We incorporate the weight matrix $\mathbf{F}$ of $u_c'$ into the current passing matrix $\mathbf{P}$ to update the accumulated weights, as detailed in \autoref{alg:weight_aggregation}. Initially, we handle the boundary cases where either matrix acts as a vector because it lacks indirect predecessors.

\noindent \textit{\underline{Case 1: $\mathbf{P}$ and $\mathbf{F}$ have no associated indirect predecessors.}}
In this scenario, since $\mathbf{P}$ has not accumulated any indirect predecessors along the propagation path (\autoref{line:p_is_vector}), both $\mathbf{P}$ and $\mathbf{F}$ act purely as weight vectors (\autoref{def:weight_vector}). Consequently, we simply compute their Hadamard product (\autoref{line:vector_mul}), identical to step \textcircled{2} of \autoref{ex:basic_counting}.

\noindent \textit{\underline{Case 2: $\mathbf{F}$ has associated indirect predecessors but $\mathbf{P}$ does not.}}
$\mathbf{F}$ functions as a weight matrix while $\mathbf{P}$ remains a weight vector (\autoref{line:f_not_vector}) under this condition. We row-scale the first block of $\mathbf{F}$ by $\mathbf{P}$, and then $\mathbf{F}$ takes over as the updated passing matrix (\autoref{line:p_scale_f}--\autoref{line:p_becomes_f}).

\noindent \textit{\underline{Case 3: $\mathbf{P}$ has associated indirect predecessors, but $\mathbf{F}$ does not.}}
$\mathbf{P}$ has already become a matrix but $\mathbf{F}$ remains a vector (\autoref{line:if_no_indirect}) in this scenario, we directly apply $\mathbf{F}$ to the first block of $\mathbf{P}$ via row-wise scaling to accumulate the candidate weights (\autoref{line:base_case_mul}).

\noindent \textit{\underline{Case 4: both $\mathbf{P}$ and $\mathbf{F}$ have associated indirect predecessors.}}
$\mathbf{P}$ and $\mathbf{F}$ both function as weight matrices (\autoref{def:weight_matrix}) in this scenario. We identify their predecessor sets $\mathcal{U}_P$ and $\mathcal{U}_F$, along with their earliest predecessor $u_{P,1}$ and $u_{F,1}$. We first process any overlapping predecessors ($u \in \mathcal{U}_P \cap \mathcal{U}_F$) by performing an element-wise Hadamard product on their corresponding blocks. This step merges their connectivity and weights simultaneously ( \autoref{line:overlap_start}--\autoref{line:overlap_end}).
A critical requirement is that only the block corresponding to the earliest indirect predecessor in the constrained sequence acts as the primary block carrying the accumulated weights. To ensure this, we compare the sequence positions of $u_{F,1}$ and $u_{P,1}$. If $u_{F,1}$ appears earlier ($u_{F,1} \prec u_{P,1}$) in the constrained sequence, $\mathbf{F}[u_{F,1}]$ becomes the new primary block. We extract the row weights $\mathbf{v}_{\text{row}}$ from the former primary block $\mathbf{P}[u_{P,1}]$, row-broadcast them into $\mathbf{F}[u_{F,1}]$, and degrade $\mathbf{P}[u_{P,1}]$ into a pure 0/1 connectivity block via binarization (\autoref{line:first_block_f}--\autoref{line:first_block_f_end}). Conversely, if $u_{P,1}$ remains the primary block, we extract the weights from $\mathbf{F}[u_{F,1}]$, apply them to $\mathbf{P}[u_{P,1}]$, and binarize $\mathbf{F}[u_{F,1}]$ (\autoref{line:first_block_p}--\autoref{line:first_block_p_end}). 
Finally, any remaining non-overlapping blocks belonging solely to the weight matrix ($u \in \mathcal{U}_F \setminus \mathcal{U}_P$) are appended to the passing matrix $\mathbf{P}$ column-wise, expanding its predecessor set accordingly (\autoref{line:append_start}--\autoref{line:append_end}).

\begin{algorithm}[t]
    \caption{\textsc{WeightAggregation}}
    \rm\sffamily
    \label{alg:weight_aggregation}
    \KwIn{Passing matrix $\mathbf{P}$ at $u_c$, Weight matrix $\mathbf{F}$ of current direct predecessor $u_c'$}
    \KwOut{Updated passing matrix $\mathbf{P}$}

    \If{$\mathbf{P}$ has no associated indirect predecessors}{ \label{line:p_is_vector}
        \If(\hfill \CommentSty{// Case 1}){$u_c'$ has no indirect predecessor}{ \label{line:if_no_indirect}
            $\mathbf{P} \gets \mathbf{P} \odot \mathbf{F}$\; \label{line:vector_mul}
            \Return{$\mathbf{P}$}
        }
        $\mathcal{U}_F \gets \text{indirect predecessors of blocks in } \mathbf{F}$ \tcp*{Case 2} \label{line:f_not_vector}
        $u_{F,1} \gets \text{the first indirect predecessor in } \mathcal{U}_F$\;
        $\mathbf{F}[u_{F,1}] \gets \mathbf{F}[u_{F,1}] \otimes_{\text{row}} \mathbf{P}$\;  \label{line:p_scale_f}
        \Return{$\mathbf{F}$} \label{line:p_becomes_f}
    }

    $\mathcal{U}_P \gets \text{indirect predecessors of blocks in } \mathbf{P}$\; \label{line:init_up}
    $u_{P,1} \gets \text{the first indirect predecessor in } \mathcal{U}_P$\;
    
    \If(\hfill \CommentSty{// Case 3}){$u_c'$ has no indirect predecessor}{ \label{line:if_no_indirect}
        $\mathbf{P}[u_{P,1}] \gets \mathbf{P}[u_{P,1}] \otimes_{\text{row}} \mathbf{F}$\;  \label{line:base_case_mul}
        \Return{$\mathbf{P}$} \label{line:base_case_return}
    }

    $\mathcal{U}_F \gets \text{indirect predecessors of blocks in } \mathbf{F}$ \tcp*{Case 4} \label{line:init_uf}
    $u_{F,1} \gets \text{the first indirect predecessor in } \mathcal{U}_F$\;
    
    \ForEach{$u \in \mathcal{U}_P \cap \mathcal{U}_F$}{ \label{line:overlap_start}
        $\mathbf{F}[u] \gets \mathbf{P}[u] \odot \mathbf{F}[u]$\;
        $\mathbf{P}[u] \gets \mathbf{F}[u]$\; \label{line:overlap_end}
    }

    \If{$u_{F,1} \prec u_{P,1}$}{ \label{line:first_block_f}
        $\mathbf{v}_{\text{row}} \gets \textsc{ExtractRowWeights}(\mathbf{P}[u_{P,1}])$\;
        $\mathbf{F}[u_{F,1}] \gets \mathbf{F}[u_{F,1}] \otimes_{\text{row}} \mathbf{v}_{\text{row}}$\;
        $\mathbf{P}[u_{P,1}] \gets \textsc{Binarize}(\mathbf{P}[u_{P,1}])$\;
        $\mathbf{P}[u_{F,1}] \gets \mathbf{F}[u_{F,1}]$\; \label{line:first_block_f_end}
    }
    \Else{ \label{line:first_block_p}
        $\mathbf{v}_{\text{row}} \gets \textsc{ExtractRowWeights}(\mathbf{F}[u_{F,1}])$\;
        $\mathbf{P}[u_{P,1}] \gets \mathbf{P}[u_{P,1}] \otimes_{\text{row}} \mathbf{v}_{\text{row}}$\;
        $\mathbf{F}[u_{F,1}] \gets \textsc{Binarize}(\mathbf{F}[u_{F,1}])$\;
        \If{$u_{F,1} \in \mathcal{U}_P$}{
            $\mathbf{P}[u_{F,1}] \gets \mathbf{F}[u_{F,1}]$\; \label{line:first_block_p_end}
        }
    }

    \ForEach{$u \in \mathcal{U}_F \setminus \mathcal{U}_P$}{ \label{line:append_start}
        $\textsc{Append}(\mathbf{P}, \mathbf{F}[u])$\;
        $\mathcal{U}_P \gets \mathcal{U}_P \cup \{u\}$\; \label{line:append_end}
    }

    \Return{$\mathbf{P}$}
\end{algorithm}

\vspace{1mm}
$\bullet$ \textit{Weight Extraction.}
During backward propagation along the constrained sequence, if the predecessor $u'_c$ recorded in the current passing matrix $\mathbf{P}$ is encountered, it signifies a \textit{cycle closure}. We perform a weight extraction operation as detailed in \autoref{alg:weight_extraction}.

At this stage, the block $\mathbf{P}[u'_c]$ forms a square matrix. Its diagonal elements indicate whether the corresponding candidates of $u'_c$ can successfully complete a cycle traversal under the indirect constraints, as in the cycle $(u_2 \rightarrow u_4 \rightarrow u_3 \rightarrow u_2)$ of $T_{Q_1}$ in \autoref{fig:sequence_counting}(a). We extract these diagonal elements into a vector $\mathbf{v}_{\text{diag}}$ (\autoref{line:extract_diag}) to capture the valid states of $u'_c$'s candidates. Since the information at $u'_c$ has now been fully evaluated and aggregated, we simplify the passing matrix by completely removing the block corresponding to $u'_c$ (\autoref{line:remove_block}). The subsequent update of $\mathbf{P}$ depends on whether other indirect predecessors remain.

\noindent \textit{\underline{Case 1: Other blocks remain in $\mathbf{P}$.}}
In this scenario (\autoref{line:check_remaining}), the presence of remaining blocks implies that $u'_c$'s block was not the primary carrier of the accumulated weights. Nevertheless, its constraint validation is crucial. Therefore, we binarize the extracted vector to retain only connectivity information (\autoref{line:binarize_diag}) and apply it via row-wise scaling to the first block of $\mathbf{P}$ (\autoref{line:row_scale_diag}). This operation effectively filters out the accumulated weights of candidates that fail to satisfy the cycle constraints.

\noindent \textit{\underline{Case 2: $u'_c$'s block was the only block in $\mathbf{P}$.}}
In this scenario (\autoref{line:no_remaining}), the extracted vector $\mathbf{v}_{\text{diag}}$ fully represents the newly accumulated weights. In this case, $\mathbf{v}_{\text{diag}}$ entirely becomes the updated passing matrix $\mathbf{P}$ (\autoref{line:p_becomes_vector}) and will continue propagating backward along the constrained sequence.

\begin{algorithm}[t]
    \caption{\textsc{WeightExtraction}}
    \rm\sffamily
    \label{alg:weight_extraction}
    \KwIn{Passing matrix $\mathbf{P}$ at $u_c$, Current direct predecessor $u_c'$}
    \KwOut{Updated passing matrix $\mathbf{P}$}

    $\mathcal{U}_P \gets \text{indirect predecessors of blocks in } \mathbf{P}$\;

    \If{$u_c' \in \mathcal{U}_P$}{ \label{line:check_uc_prime}
        $\mathbf{v}_{\text{diag}} \gets \textsc{ExtractDiag}(\mathbf{P}[u_c'])$\; \label{line:extract_diag}
        $\mathcal{U}_P \gets \mathcal{U}_P \setminus \{u_c'\}$\;
        $\textsc{RemoveBlock}(\mathbf{P}, u_c')$\; \label{line:remove_block}

        \If(\hfill \CommentSty{// Case 1}){$\mathcal{U}_P \neq \emptyset$}{ \label{line:check_remaining}
            $u_{P,1} \gets \text{the first indirect predecessor in } \mathcal{U}_P$\;
            $\mathbf{v}_{\text{diag}} \gets \textsc{Binarize}(\mathbf{v}_{\text{diag}})$\; \label{line:binarize_diag}
            $\mathbf{P}[u_{P,1}] \gets \mathbf{P}[u_{P,1}] \otimes_{\text{row}} \mathbf{v}_{\text{diag}}$\; \label{line:row_scale_diag}
        }
        \Else(\hfill \CommentSty{// Case 2}){ \label{line:no_remaining}
            $\mathbf{P} \gets \mathbf{v}_{\text{diag}}$\; \label{line:p_becomes_vector}
        }
    }
    
    \Return{$\mathbf{P}$} \label{line:return_p}
\end{algorithm}

\Cref{alg:sequence_counting} outlines the overall procedure for computing the candidate weight $W(u, v)$ under both sequence constraints and standard tree constraints. The algorithm performs a backward propagation along the constrained sequence $s$.

We initialize the process at the last vertex $u_c$ of sequence $s$, setting its weight matrix as the initial state of the passing matrix $\mathbf{P}$ (\autoref{line:init_uc}--\autoref{line:init_p}). During the reversed traversal of $s$, for each adjacent pair $(u_c', u_c)$, we first propagate the passing matrix $\mathbf{P}$ from $u_c$ to its direct predecessor $u'_c$ by multiplying it with the transition matrix $\mathbf{A}_{u'_c, u_c}$ (\autoref{line:matrix_mult}). Then, we generate the weight matrix $\mathbf{F}$ for $u'_c$ (\autoref{line:get_weight}). If $u'_c$ closes a cycle—meaning it is an indirect predecessor previously visited—we refine the passing matrix via the \textsc{WeightExtraction} operation to capture cycle constraints (\autoref{line:extraction}). Subsequently, the state is updated by fusing $\mathbf{P}$ and $\mathbf{F}$ via the \textsc{WeightAggregation} operation (\autoref{line:aggregation}).

Once the propagation reaches the head vertex of $s$, the finalized passing matrix $\mathbf{P}$ degrades into a vector containing the accumulated valid weights. We sum all elements of this vector to obtain the total sequence-constrained weight $w_{\text{con}}$ (\autoref{line:sum_con}). Furthermore, the query vertex $u$ may also possess isolated children that are completely independent of any Type-I non-tree constraints. We compute their normal candidate weight $w_{\text{norm}}$ using \autoref{eq:tree_counting} (\autoref{line:sum_norm}). Finally, the product of $w_{\text{con}}$ and $w_{\text{norm}}$ yields the final matching weight $W(u, v)$ (\autoref{line:final_weight}).

\begin{algorithm}[t]
    \caption{\textsc{SequenceCounting}}
    \rm\sffamily
    \label{alg:sequence_counting}
    \KwIn{$u, v$, $CS$, $s$}
    \KwOut{$W(u, v)$}

    $u_c \gets \text{the last vertex in } s$\; \label{line:init_uc}
    $\mathbf{P} \gets \textsc{WeightMatrix}(u_c)$\; \label{line:init_p}

    \ForEach{$(u_c', u_c)$ in the reversed traversal of $s$}{ \label{line:backward_loop}
        $\mathbf{P} \gets \mathbf{A}_{u_c' u_c} \mathbf{P}$\; \label{line:matrix_mult}
        $\mathbf{F} \gets \textsc{WeightMatrix}(u_c')$\; \label{line:get_weight}

        \If{$u_c'$ closes a cycle}{ \label{line:check_cycle}
            $\mathbf{P} \gets \textsc{WeightExtraction}(\mathbf{P}, u_c')$\; \label{line:extraction}
        }

        $\mathbf{P} \gets \textsc{WeightAggregation}(\mathbf{P}, \mathbf{F}, u_c')$\; \label{line:aggregation}
    }

    $w_{\text{con}} \gets \textsc{Sum}(\mathbf{P})$\; \label{line:sum_con}
    $w_{\text{norm}} \gets \textsc{IsolatedChildrenWeight}(u, v)$\; \label{line:sum_norm}
    $W(u, v) \gets w_{\text{con}} \times w_{\text{norm}}$\; \label{line:final_weight}

    \Return{$W(u, v)$}
\end{algorithm}

\begin{example}
    Continuing \autoref{ex:updated_initial_feature}, where the weight matrix $\mathbf{F}_4$ is constructed in \autoref{fig:sequence_counting}(h). The propagation process is illustrated in \autoref{fig:sequence_counting}(i). Starting from $u_4$, the weight matrix $\mathbf{F}_4$is propagated to $u_3$ by multiplication with $\mathbf{A}_{34}$, illustrated in Step \textcircled{1}. The weight vector of $u_3$, which encodes the weights of its candidates, is then incorporated into the first block of $\mathbf{P}$ through row-broadcast as specified in \autoref{line:base_case_mul} of \autoref{alg:weight_aggregation}, producing the updated passing matrix in Step \textcircled{2}. 
    The passing matrix is subsequently propagated backward to $u_2$ through $\mathbf{A}_{23}$, obtaining $\mathbf{P}$ in \textcircled{3}. Since $u_2$ is the indirect predecessor of $u_4$, whose constraint on $u_4$ is enforced in the passing matrix, we extract the block of $u_2$. 
    In Step \textcircled{4}, the diagonal (circled in red) of the block corresponding to $u_2$ is extracted, obtaining $\mathbf{V}_{diag}$ (\autoref{line:extract_diag} of \autoref{alg:weight_extraction}), and then this block is removed from $\mathbf{P}$ (\autoref{line:remove_block} of \autoref{alg:weight_extraction}). $\mathbf{V}_{diag}$ is first binarized (\autoref{line:binarize_diag} of \autoref{alg:weight_extraction}) and applied to the first block of $\mathbf{P}$ (\autoref{line:row_scale_diag} of \autoref{alg:weight_extraction}), yielding the updated $\mathbf{P}$ after weight extraction as shown in Step \textcircled{5}. In Step \textcircled{6}, the initial weight vector $\mathbf{V}_2$ of $u_2$ is row-scaled into $\mathbf{P}$, which completes the weight aggregation at $u_2$ and produces the updated $\mathbf{P}$.
    After that, the obtained $\mathbf{P}$ continues to propagate backward to $u_1$ through $\mathbf{A}_{12}$. At this point, 
    $u_1$ is the only remaining indirect predecessor encoded in the passing matrix, the diagonal (circled in red) of $\mathbf{P}$ is extracted to form the vector $\mathbf{V}_{diag}$ as shown in Step \textcircled{7}, which is then combined with $\mathbf{V}_1$ via a Hadamard product.
    Finally, summing the entries of $\mathbf{P}$ in Step \textcircled{8} yields the constrained sequence count, which equals to $2$. This value is exactly the true homomorphism count of the constrained sequence.
    
\end{example}

\noindent \textbf{Remark}.
It should be remarked that the homomorphism counts obtained from the constrained sequence counting are not exact. When the constrained sequence contains indirect constraints, multiple indirect constraints can simultaneously affect a single vertex. Although the regularized weight matrix is introduced to precompute such joint constraint effects before propagation, it can only approximate the proportion of valid weights propagated based on the connected paths. In fact, the valid weight contributed by each connected path is different. The counting rules for such a complicated structure might be further optimized in future work.

\subsection{Complexity Analysis}

\begin{theorem}
    Given a connected component $S$ under a query vertex in the query tree, the time complexity of the overall counting process for $S$ is polynomial to the size of $S$ and the candidate space.
\end{theorem}

\begin{proof}
In the counting process, we only conduct two rounds of matrix propagation. At each step of the propagation, we conduct a constant number of matrix operations (i.e., matrix multiplication), which can be done in cubic time of the candidate size of query vertices.
Since the number of steps in each propagation is bounded by the query graph size, the theorem is immediate.
\end{proof}

To assess the efficiency of the constrained sequence counting, we conduct a preliminary experiment on a synthetic dataset. Specifically, we use a $32$-clique query graph and a $32$-clique data graph, where vertices share the same label.
By definition, the exact count is $32!$. We evaluate this instance using two representative subgraph matching algorithms recommended in recent survey papers~\cite{2024survey, RapidExperiment}.
None of them can finish within two hours. The state-of-the-art sampling-based subgraph counting method \fastest ~\cite{Fastest} also fails to return a valid result within $100$ seconds. In contrast, our method computes the count in $7.11$ seconds with a \qerror of $1.98$, where the majority of the runtime is spent on candidate space construction and refinement, and the counting phase only spends $0.52$ seconds.

\section{Handling Type-II Non-tree Edge Constraints}
\label{sec:lca}
In this section, we introduce the \textit{lowest common ancestor} (LCA for short) to handle the Type-II non-tree edges constraints.

Recall that a Type-II non-tree edge connects two vertices with different parents in the query tree.
However, the two vertices must have an LCA vertex in the tree. Clearly, the paths from the LCA vertex to the two endpoints form a triangle together with this non-tree edge. 
A necessary condition for a valid match of the query graph is that such a triangle must also exist in the candidate space.
However, in the candidate tree-based counting method defined by \autoref{eq:tree_counting}, the weights of candidate vertices are propagated upward along the tree paths without considering such non-tree edge constraints.
To remove the false positive counting, we develop an LCA ratio computation method. Given a Type-II non-tree edge, we check if the two endpoints of its matching edge can reach the same LCA candidate. By calculating the positive ratio, we can therefore estimate the valid subgraph homomorphism count. 

To compute the LCA ratio, we first find the descendants for each candidate of an LCA vertex using a dynamic programming method. Then, based on the LCA descendants, we compute the LCA ratio for each candidate by using set intersections to count descendant pairs that form valid triangles across endpoints.

\subsection{Finding LCA Descendants}
\label{sec:lca_descendants}
For each LCA candidate, we compute its descendant candidates located at the endpoints of a non-tree edge. 
This is done by propagating reachable candidates downward along the tree path from the LCA to the endpoint. 
To speed up the computation, we adopt a level-by-level manner where the results computed in the previous level can be reused in the current level. This can be implemented as a dynamic programming on the candidate tree. 

\vspace{1mm}
\noindent \textbf{Algorithm Details}.
The computation details are summarized in \autoref{alg:LCA_Descendants}. For each Type-II non-tree edge, we form two LCA paths, each connecting the LCA to one of its endpoints.
We perform a one-pass top-down traversal of the query tree $T_Q$ in BFS order, and compute the descendants of all LCA paths simultaneously. For each LCA path $P$, we maintain a path state $Current\_descendants[P]$, which functions as a mapping from the LCA candidates to their descendants at the current query vertex. As we traverse downwards, the path states at the parent $u_p$ are propagated to the current vertex $u$. Specifically, when visiting $u$, we first identify all $v_p \in C(u_p)$ involved in any LCA path passing through its parent $u_p$ (\autoref{line:union_parents}) and cache their child candidates $C(u \mid u_p, v_p)$ (\autoref{line:cache_children}) to ensure each candidate edge is processed only once. Subsequently, for each path $P$ passing through $u$, we update the path state $Current\_descendants[P]$ by replacing the parent candidates $v_p$ with their cached child candidates $Cache[v_p]$ (\autoref{line:extend_paths}). If $u$ is the LCA vertex of a path $P$, its state is initialized as an identity mapping of $C(u)$ (\autoref{line:initial_state}). Once $u$ reaches the endpoint of a path $P$, the current state is saved (\autoref{line:save_state}) as the mapping from the LCA candidates to the endpoint descendants.

\begin{lemma}
    The complexity of Algorithm~\ref{alg:LCA_Descendants} is $O(|E_C| + |V_Q||E_Q||V_C|)$, where $V_Q$ and $E_Q$ are the query graph vertex and edge sets, and $V_C$ and $E_C$ are the vertex set and edge set in the candidate space. 
\end{lemma}

\begin{proof}
    The process of finding LCA descendants traverses the candidate tree in a top-down manner. By employing the cache mechanism, each candidate edge is scanned only once, resulting in a complexity of $O(|E_C|)$. During the traversal, the algorithm propagates path states at each $u$. The number of LCA paths passing through a single vertex $u$ is proportional to the tree depth $O(|E_Q|)$. For each LCA path, the path state maintains a mapping from LCA candidates to the candidates of the current vertex $u$, where the total number of involved candidates is $O(|V_C|)$. Therefore, the complexity of path states propagation is $(|V_Q||E_Q||V_C|)$. The overall complexity of the algorithm is $O(|E_C| + |V_Q||E_Q||V_C|)$.
\end{proof}

\vspace{1mm}
\noindent \textbf{Discussion}.
The high efficiency of Algorithm~\ref{alg:LCA_Descendants} is achieved by reusing the intermediate results as follows.  

$\bullet$ For two non-tree edges sharing the same LCA-endpoint path, the descendant results along a shared path only need to be computed once because the endpoints of a non-tree edge are decomposed.
    
$\bullet$ For two LCA-endpoint paths sharing the LCA, if one endpoint is the ancestor of the other, the descendant candidates computed for the ancestor endpoint can be reused for the other.
    
$\bullet$ In one LCA-endpoint path, multiple LCA candidates may map to the same descendant candidate at an intermediate vertex. The computation of that descendant candidate's child candidates is performed only once and shared across all such LCA candidates.

\begin{algorithm}[t]
    \caption{\textsc{Find LCA Descendants}}
    \rm\sffamily
    \label{alg:LCA_Descendants}
    \KwIn{The query tree $T_Q$, corresponding $CS$}
    \KwOut{LCA descendants for each LCA path}
    \SetKwProg{Fn}{Function}{:}{}

    \ForEach{$u$ in BFS Traversal of $T_Q$}{
        \If{$u$ is involved in any LCA path} {
            $u_p \gets$ Parent($u$)\;
            \If{$u_p$ exists} {
                
                All\_current\_descendants $\gets \{v_p \mid v_p \in Current\_descendants[P], \forall P \text{ covering }u_p\}$\;
                \label{line:union_parents}
                
                \ForEach{$v_p \in $ All\_current\_descendants}{Cache[$v_p$] $= C(u \mid u_p, v_p)$\;}
                \label{line:cache_children}
                
                \ForEach{LCA path $P$ passing through $u$}{
                    Current\_descendants[$P$] $\gets \{Cache[v_p] \mid v_p \in Current\_descendants[P]\}$\;
                    \label{line:extend_paths}
                }
                
            }
            \If{$u$ is the LCA of path $P$}{
                Current\_descendants[$P$] $\gets C(u)$\; 
                \label{line:initial_state}
                
            }
    
            \If{$u$ is the Endpoint of path $P$}{
                Save Current\_descendants[$P$]\;
                \label{line:save_state}
            }            
        }
    }
\end{algorithm}

\subsection{Computing LCA Ratio}
We view the tree paths from the LCA to the two non-tree endpoints as virtual edges. Together with the non-tree edge, these two virtual edges form a triangle. Our goal is to count, for each LCA candidate, how many combinations of endpoint candidates form valid triangles. Specifically, for a given LCA candidate, we first extract the candidate sets of the two endpoints of a Type-II non-tree edge that connected to it through the candidate tree paths. The total number of possible triangles is given by the product of the sizes of these two candidate sets, which aligns with the counting logic of \autoref{eq:tree_counting}. We then explicitly examine how many candidate edges exist between the two endpoint candidate sets. Each such candidate edge, together with the connections from its endpoints to the LCA candidate, forms a valid triangle.

A single LCA vertex may be associated with multiple non-tree edges, some of which share a common endpoint. We group such non-tree edges by their shared endpoint and treat this endpoint as an anchor. Within each group, we fix one candidate path from the LCA to the anchor, and compute the fraction of candidate combinations of the other endpoints that satisfy all triangle constraints induced by the group. The final LCA ratio for an LCA candidate is obtained by multiplying the ratios over all anchor groups.

\begin{lemma}
    Given an LCA-non-tree triangle, the LCA ratio can be computed in $O(|E_C||V_C|)$ time.
\end{lemma}

\begin{proof}
    For each candidate edge between the LCA and the anchor endpoint, we compute the intersection of the candidate neighbors of the LCA candidate and the anchor candidate.
    We conduct such intersections for $O(|E_C|)$ anchored candidate edges, each consuming $O(|V_C|)$ time. Thus, the lemma is complete.
\end{proof}

\vspace{1mm}
\noindent \textbf{Remark}.
By integrating Type-I and Type-II non-tree constraints with the existing tree-edge constraints in the candidate tree framework, our structural module bridges the gap between tree homomorphism and graph homomorphism, extending the candidate tree framework to full subgraph homomorphism counting. This extension is non-trivial, since tree homomorphism counting is polynomial-time solvable, whereas subgraph homomorphism counting itself is a \#P-hard problem. Compared with \fastest, which relies on sampling over simple tree homomorphism counts, our approach moves the more complex structural-constraint computation into the counting stage and leaves only the injectivity constraints to the sampling module. For queries that do not involve injectivity violations, the sampling stage can even be skipped entirely.

\section{Local Sampling}
\label{sec:iso}
While the non-tree edge constraints can be processed by the matrix computation-based method, the isomorphism constraint remains unsolved. This is because the candidate tree-based counting only propagates and aggregates vertex weights, rather than the specific matching states. In this section, we introduce a local sampling approach to estimate the proportion of isomorphism count within the homomorphism count. Unlike previous sampling-based methods, which usually conduct sampling over the entire candidate space, we conduct sampling in a local candidate space, where isomorphism constraints arise. In particular, isomorphism violations can only occur among candidates of the same-label query vertices; even within this set, clique-structured matches cannot violate injectivity, since no self-loops exist.

Query vertices with the same label may appear at different positions in the query graph and are not necessarily adjacent. Extracting them in isolation would break the structural connectivity constraints among them. To preserve these inherent structural constraints, we directly identify a connected subgraph within the original query graph that encompasses all these same-label vertices. This connected subgraph naturally consists of both tree edges and non-tree edges from the query graph, seamlessly maintaining the exact topological relationships. As a result, the target same-label vertices form a locally connected structure where each vertex is properly constrained by its neighbors. We then perform sampling exclusively within the candidate space induced by this locally connected subgraph to estimate the isomorphism ratio.

The sampling procedure proceeds by following the traversal order of the connected subgraph. When sampling a query vertex, its feasible sampling region is bounded by its already visited neighbors, meaning that earlier sampled vertices directly constrain the subsequently sampled ones. To optimize this process, we apply a greedy selection rule: among its already visited neighbors, we select the one with the strongest constraint as its tree parent to induce the feasible sampling region, from which a candidate vertex is then uniformly sampled, with the inverse sampling probability accumulated throughout the sampling process. The pre-computed homomorphism weights $W(u, v)$ from the prior counting phase, used in a similar way to \cite{Fastest}, are utilized to accurately reflect the matching distribution and continuously monitor the statistical variance of the random walks. The isomorphism ratio is then estimated as the weighted fraction of isomorphic samples among the homomorphic ones. For queries where the empirical samples fail to reach statistical confidence, we approximate this ratio by extrapolating the joint constraint-satisfaction probabilities for homomorphic and isomorphic samples separately from their average marginal survival rates. Additionally, a weight cutoff strategy is applied to truncate overly large sample weights, mitigating extreme variance caused by overly dense homomorphic regions.

\section{Experiment}
\subsection{Experiment Setup}
\noindent \textbf{Datasets}.
We conduct experiments on $10$ real-world datasets.
The first $8$ datasets have publicly available query sets and are widely used in previous works like~\cite{flowsc, Fastest, NeurSC}, 
while the remaining $2$ are large-scale graphs where we generate queries via random walks on data graphs following~\cite{flowsc, RapidExperiment}. To obtain the exact subgraph counts for each query, we use the recommended method in~\cite{RapidExperiment}. Given the extreme difficulty of subgraph counting, we set a timeout of $30$ minutes~\cite{NeurSC, flowsc} for computing the ground-truth counts on the first $8$ datasets, and a timeout of $2$ hours~\cite{flowsc} for the two large-scale graphs. Our experiments are evaluated on queries with obtainable ground-truth counts, which are further grouped into subsets by size for evaluation.
Dataset statistics are reported in \autoref{tab:dataset}.

\begin{table}[t]
    \caption{Statistics of Datasets. } 
    \label{tab:dataset}
    \centering
    \resizebox{0.9\linewidth}{!}{
        \begin{tabular}{l r r r r r}
        \toprule
        Dataset & $|V_G|$ & $|E_G|$ & $|\Sigma|$ & $d$ & $|V_Q|$ \\ \midrule
        \Yeast & 3,112 & 12,519 & 71 & 9.0 & 4 to 32 \\
        \Hprd & 9,460 & 34,998 & 307 & 7.4 & 4 to 32 \\
        \Human & 4,674 & 86,282 & 44 & 36.9 & 4 to 20 \\
        \Wordnet & 76,853 & 120,399 & 5 & 3.1 & 4 to 20 \\
        \Dblp & 317,080 & 1,049,866 & 15 & 6.6 & 4 to 16 \\
        \Youtube & 1,134,890 & 2,987,624  & 25  & 5.3 & 4 to 32\\
        \EU & 862,664 & 16,138,468 & 40 & 37.4 & 4 to 8 \\ 
        \Patents & 3,774,768 & 16,518,947 & 20 & 8.8 & 4 to 32 \\ 
        \Twitter & 41,652,230 & 1,202,513,344 & 100 & 57.7 & 4 \\
        \Friendster & 65,608,366 & 1,806,067,135 & 100 & 55.1 & 4 to 8 \\
        \bottomrule
        \end{tabular}
    }
\end{table}

\vspace{1mm}
\noindent \textbf{Compared Algorithms}.
We compare \asc with the following approximate algorithms: (1) \flowsc\cite{flowsc}, the state-of-the-art learning-based method; (2) \fastest\cite{Fastest}, the state-of-the-art sampling-based method; (3) \learnsc\cite{Learnsc}; (4) \neursc\cite{NeurSC}; and (5) \lss\cite{LSS}. 
We also evaluate subgraph homomorphism counting against three baselines: \fastest~\cite{Fastest}, \alley~\cite{Alley}, and its index-optimized variant \alleytpi~\cite{Alley}.
The source code of \flowsc, \fastest, \neursc, \alley, \alleytpi, and \lss are publicly available, and \learnsc is obtained from the authors. All parameter settings follow their default values. 

\vspace{1mm}
\noindent \textbf{Query Setting}.
Following the common setting~\cite{RapidExperiment, flowsc, NeurSC}, we load a data graph and a query graph into memory each time, and report the elapsed time per data-query pair and the \qerror$ = \max(\frac{\max(1, c)}{\max(1, \hat{c})}, \frac{\max(1, \hat{c})}{\max(1, c)})$, where $\hat{c}$ and $c$ are the estimated and ground truth counts, respectively. Experiments are conducted on a Ubuntu $22.04.1$ system, equipped with an Intel Xeon Silver $4314$ CPU @ $2.40$GHz with $64$ cores.

\subsection{Performance Evaluation}
\begin{figure*}[t]
    \centering
    \includegraphics[width=\textwidth]{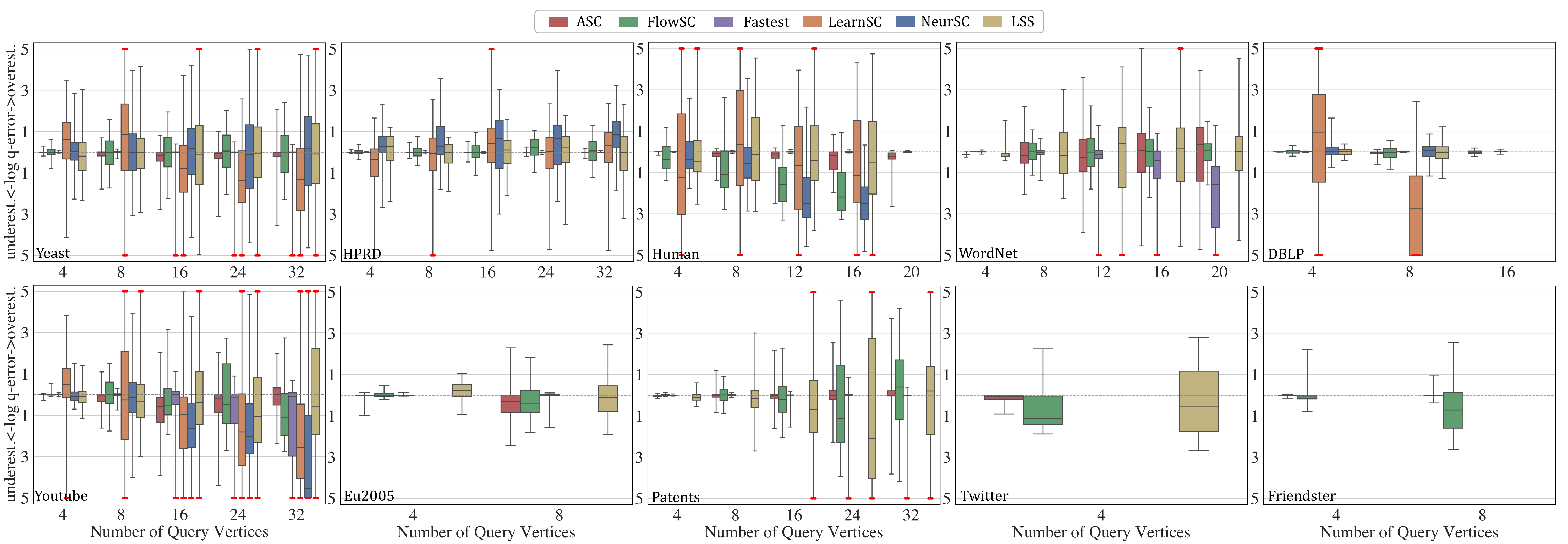}
    \caption{Evaluating accuracy with the x-axis representing the number of query graph vertices.} 
    \label{fig:accuracy}
\end{figure*}

\noindent \textbf{Accuracy Comparison}.
\autoref{fig:accuracy} reports the accuracy performance over all datasets. The boxplot presents the distribution of log \qerror of a query graph set. The upper bound (Resp. lower bound) of a box is $75\%$ (Resp. $25\%$) quantile of the \qerror, the whiskers contain all \qerror{}s from minimum to maximum, and the line in the box is the median. Values exceeding $10^5$ are clipped, with bold red caps indicating truncated whiskers.

Overall, \asc demonstrates consistently higher and more stable accuracy. 
On the small dataset \Yeast, \asc achieves stable performance across query sets of different sizes, attaining the lowest average \qerror among all methods and yielding a $3.24\times$ improvement over the state-of-the-art method \flowsc.
Although the sampling-based method \fastest performs nearly perfectly on small queries in \Yeast, it begins to suffer from sampling failures when the query size reaches $16$, resulting in a sharp degradation in accuracy. On the other two small datasets, \Hprd and \Human, \asc outperforms all learning-based methods, but is slightly worse than \fastest, as \fastest is particularly effective when the sampling space is small. On large datasets such as \Youtube and \Patents, \fastest performs well only for small query sizes. In contrast, \asc retains high accuracy on larger datasets and for larger queries. For example, on \Youtube{}\_32, \asc achieves the best performance with a low \qerror, while \flowsc degrades to a substantially higher error and \fastest fails entirely, yielding a \qerror about $7$ orders of magnitude larger than that of \asc.

In terms of learning-based methods, \neursc and \learnsc are difficult to train on large datasets.
Although \lss can handle \Twitter, it is significantly outperformed by \asc.
\flowsc can process all datasets. While it achieves higher accuracy than other learning-based methods, but returns lower accuracy than \asc on most datasets, except for \Wordnet. The exception is largely due to the low variance of the true counts in \Wordnet, where many queries share identical counts. Such distributions favor regression-based learning methods because of their tendency to regress toward the mean~\cite{flowsc}. Consequently, while \learnsc and \neursc do not benefit from this property due to their difficulty in training, \flowsc and \lss perform particularly well on \Wordnet. 

In addition, these learning-based methods rely heavily on supervision, where the performance is greatly influenced by the number of training samples.
As a result, even \flowsc, which performs best among learning-based methods, obtains poor accuracy on the \Human dataset. Moreover, on \Human{}\_20 and \Dblp{}\_16, the scarcity of effective training samples prevents the models from learning. In contrast, the performance of \asc and \fastest is not constrained by the number of training samples, and both methods perform well even on datasets with limited samples.
Moreover, only \asc and \flowsc can scale to the billion-scale dataset, where \asc still achieves higher accuracy than \flowsc.

\vspace{1mm}
\noindent \textbf{Counting Range Analysis}.
We conduct an experiment on a selected dataset to evaluate how the ground truth count of query graphs influences the distribution of \qerror. As shown in \autoref{fig:counting_range}, our algorithm performs consistently across different count ranges. In the range with larger counts, the accuracy becomes slightly worse, mainly due to the scale effect: queries with larger candidate spaces naturally accumulate more approximation error. The learning-based methods \flowsc, \learnsc, and \neursc tend to overestimate small-count queries, and underestimate large-count queries, which is consistent with the regression-to-the-mean effect commonly observed in regression-based approaches~\cite{flowsc}. The active learning strategy of \lss alleviates this issue. The accuracy of \fastest does not show a clear pattern across different counting ranges. Its performance depends on sampling difficulty, and sampling failures are observed in every interval from $10^2$ to $10^{10}$.

\begin{figure*}[t]
    \centering
    \begin{minipage}{0.25\linewidth}
        \centering
        \includegraphics[width=\linewidth]{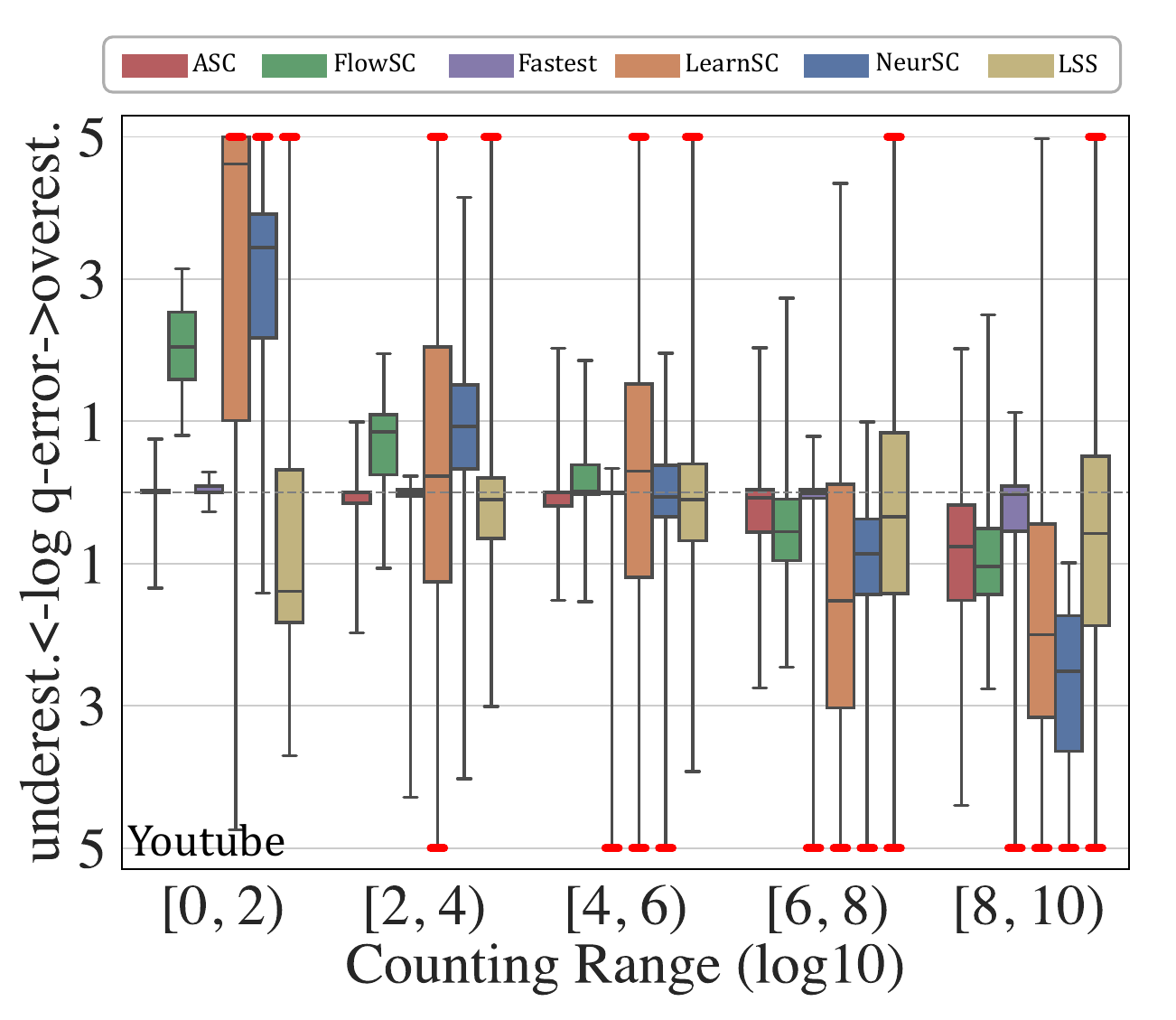}
        \vspace{-7mm}
        \caption{Effect of counting range.}
        \label{fig:counting_range}
    \end{minipage}
    \hspace{-3mm}
    \begin{minipage}{0.36\linewidth}
        \centering
        \vspace{1mm}
        \includegraphics[width=\linewidth]{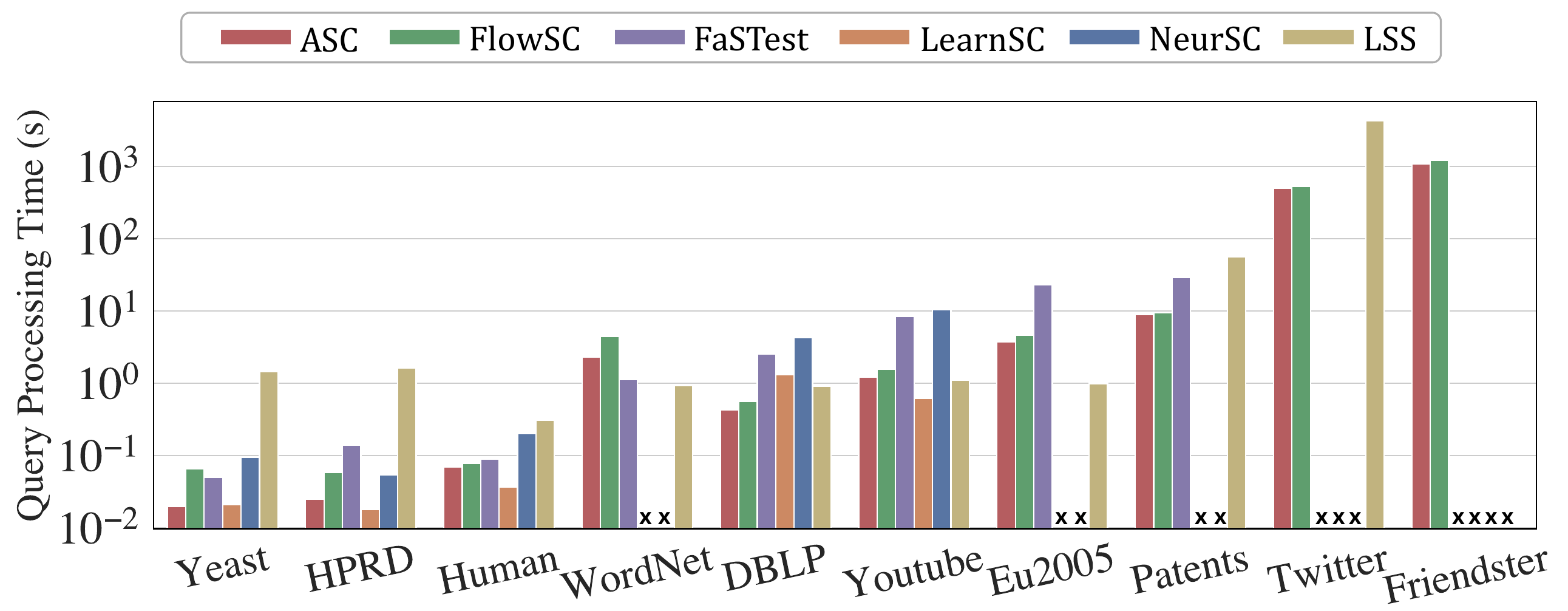} 
        \vspace{2mm}
        \caption{Efficiency Performance.}
        \label{fig:query_time}
    \end{minipage}
    \hspace{-6mm}
    \begin{minipage}{0.42\linewidth}
        \centering
        \captionof{table}{Homomorphism evaluation.\label{tab:homo_acc}}
        \resizebox{0.9\linewidth}{!}{
            \begin{tabular}{l|r r | r r | r r | r r}
            \toprule
            \multirow{2}{*}{\diagbox{Data.}{Alg.}} & \multicolumn{2}{c|}{\asc} & \multicolumn{2}{c|}{\fastest} & \multicolumn{2}{c|}{\alley} & \multicolumn{2}{c}{\alleytpi} \\
            \cmidrule{2-9}
            & q-error & Time & q-error & Time & q-error & Time & q-error & Time \\ \midrule
            \Yeast & 1.05 & 0.02 & 1.02 & 0.04 & 1$\times10^4$ & 0.17 & 1$\times10^3$ & 0.17 \\
            \Hprd & 1.01 & 0.03 & 1.00 & 0.11 & 22.45 & 0.01 & 13.15 & 0.01 \\
            \Human & 1.18 & 0.08 & 1.01 & 0.09 & 7$\times 10^8$ &6$\times10^4$& 7$\times10^8$ & 6$\times10^4$ \\
            \Wordnet & 1.00 & 2.23 & 1.08 & 0.25 & 11.85 & 7$\times10^3$ & 5.90 & 4$\times10^4$ \\
            \Dblp & 1.09 & 0.46 & 1.01 &1.64 & - & - & - & - \\
            \Youtube & 2.74 & 1.32 & 1.32 & 6.10 & - & - & - & -\\
            \EU & 6.48 & 3.92 & 1.14 & 16.69 & - & - & - & -\\
            \Patents & 1.22 & 8.83 & 1.05 & 38.52 & - & - & - & - \\
            \Twitter & 1.04 & 505.05 &- &- & - & -& -& - \\
            \Friendster & 1.15 & 1$\times10^3$ & -&- & -&- & -&- \\ \bottomrule
            \end{tabular} 
        }
    \end{minipage}

\end{figure*}

\vspace{1mm}
\noindent \textbf{Query Processing Time}.
\autoref{fig:query_time} reports the average processing time. 
A cross indicates that a method encounters training failure or out-of-memory. In general, \asc is highly competitive on both small and large datasets. On small datasets such as \Yeast, \Hprd, and \Human, \asc is faster than \fastest and achieves runtime comparable to the fastest learning-based method, \flowsc. 
\asc first incorporates structural constraints into a polynomial-time counting stage, which is efficient in practice. Its sampling stage is then performed only on local subspaces where injectivity violations may occur. If a query does not induce any injectivity violations, the sampling stage can be skipped entirely. \fastest performs sampling based on tree homomorphisms, its sampling is conducted over the full candidate space, which is often much larger than the solution space, making isomorphic samples difficult to obtain. For example, in our preliminary experiment of counting a same-label $32$-clique in a same-label 32-clique data graph, \fastest did not return a result within $100$ seconds, while \asc skipped sampling entirely because this query does not produce injectivity violations.
On the \Wordnet dataset, \neursc and \learnsc cannot be applied due to excessive training time. Both \asc and \flowsc are slower than \fastest because their filtering technique \bp is less efficient on graphs with highly concentrated label distributions. On larger datasets such as \Youtube, \EU, and \Patents, \asc and \flowsc exhibit similar runtime and show a clear advantage over \fastest. As the graph size increases, the full candidate space becomes much larger, causing \fastest to spend more time obtaining isomorphic samples. 
Due to their complex network architectures, \neursc and \learnsc cannot scale to large datasets. Although \lss uses a lightweight model and is therefore more scalable, it still cannot process the largest dataset, \Friendster. Among learning-based approaches, \flowsc offers the best efficiency and scalability, mainly due to its highly efficient one-pass bottom-up flow-learning model. Overall, \asc achieves efficiency comparable to \flowsc on large and billion-scale datasets such as \Twitter and \Friendster.

\subsection{Testing on Homomorphism}
We conduct experiments for subgraph homomorphism counting to evaluate the effectiveness of the main counting module, i.e., matrix computation-based counting (\autoref{sec:Module1}) together with the LCA ratio (\autoref{sec:lca}). Since computing the exact subgraph homomorphism counts is also computationally expensive, we select those queries with existing subgraph isomorphism ground truth that do not induce isomorphism constraints. Specifically, a query graph will not violate isomorphism constraints if its vertex labels are unique, or if vertices with the same label form a clique. In such cases, the subgraph isomorphism count coincides with the subgraph homomorphism count. We select all queries satisfying these conditions from each dataset for homomorphism testing.
We compare \asc with \fastest and the homomorphism counting method \alley and its variant \alleytpi. 

As shown in \Cref{tab:homo_acc}, \asc provides a better accuracy-efficiency trade-off, it achieves accuracy comparable to \fastest with substantially faster runtime, and clearly outperforms \alley and \alleytpi in both accuracy and efficiency. \alley and \alleytpi frequently suffer from sampling failures, which leads to unsatisfactory performance even on small datasets. In addition, \alleytpi incurs extra overhead for index construction, making it unsuitable for random data-query pairs. We omit their results on larger datasets. \fastest achieves high accuracy but at the cost of substantial runtime. In contrast, \asc attains accuracy comparable to \fastest on most datasets while requiring significantly less time. Although \asc is less accurate than \fastest on \Youtube and \EU, where the candidate space is large, its accuracy is acceptable given the substantially lower runtime. Although \Wordnet also has a larger candidate space, it contains only four $4-$size queries without isomorphism constraints, and \asc achieves near-zero error on these queries.

\subsection{Evaluation of Individual Techniques}

\noindent \textbf{Error Characterization of Type-I Module.}
We conduct error characterization experiments on small datasets \Yeast and \Human to analyze the factors that affect algebraic counting. 
However, error characterization requires exact counting of graph homomorphisms, which is extremely time-consuming even on small datasets. To tackle the computational difficulty, we reduce the exact counting to the child neighborhood of a single candidate $v$, as defined in \autoref{eq:counting_unit}, and compare $W(u, v)$ with the approximate count produced by the algebraic method.
To better investigate the underlying characteristics of our method, we consider $4$ scale-insensitive high-order features of the constrained sequences: edge density (the number of internal edges/sequence length), indirect ratio (the number of indirect edges/internal edges), conflict intensity (the number of nested or intersecting cycles/indirect edge), and norm conflict (the number of nested or intersecting cycles/sequence length). The correlations are quantified using Spearman's Rho as shown in \autoref{fig:EC_T1}. As observed, the norm conflict and conflict intensity are the most significant factors of \qerror on \Yeast and \Human, respectively. Since the two factors are highly correlated, the results suggest that the conflict pairs (i.e., nested or intersecting indirect edges within the sequence) are the main error source in algebraic counting.

We also observe that the strongest correlation with \qerror reaches around $0.5$ on \Yeast and $0.8$ on \Human. This is because our analysis focuses on the structural features of constrained sequences in the query graph, while the corresponding candidate space also substantially affects accuracy.  
Interestingly, it is observed that, for a sequence $s$ under $u$, even if we cannot guarantee the exact solution due to nested or indirect edges in $s$, the weights $W(u, v)$ for some candidates $v$ may still achieve perfect accuracy with \qerror $=1$.
This is attributed to the effort that we introduce a simple penalty based on precomputed reachable weight ratios along the reverse sequence, which can suppress the accumulation of invalid weights during propagation. 
Although this penalty cannot guarantee counting that exactly conforms to the matching states in the candidate space, it performs very well in practice, which is frequently observed in our experiments.

\begin{figure}[t]
    \centering
    \includegraphics[width=\linewidth]{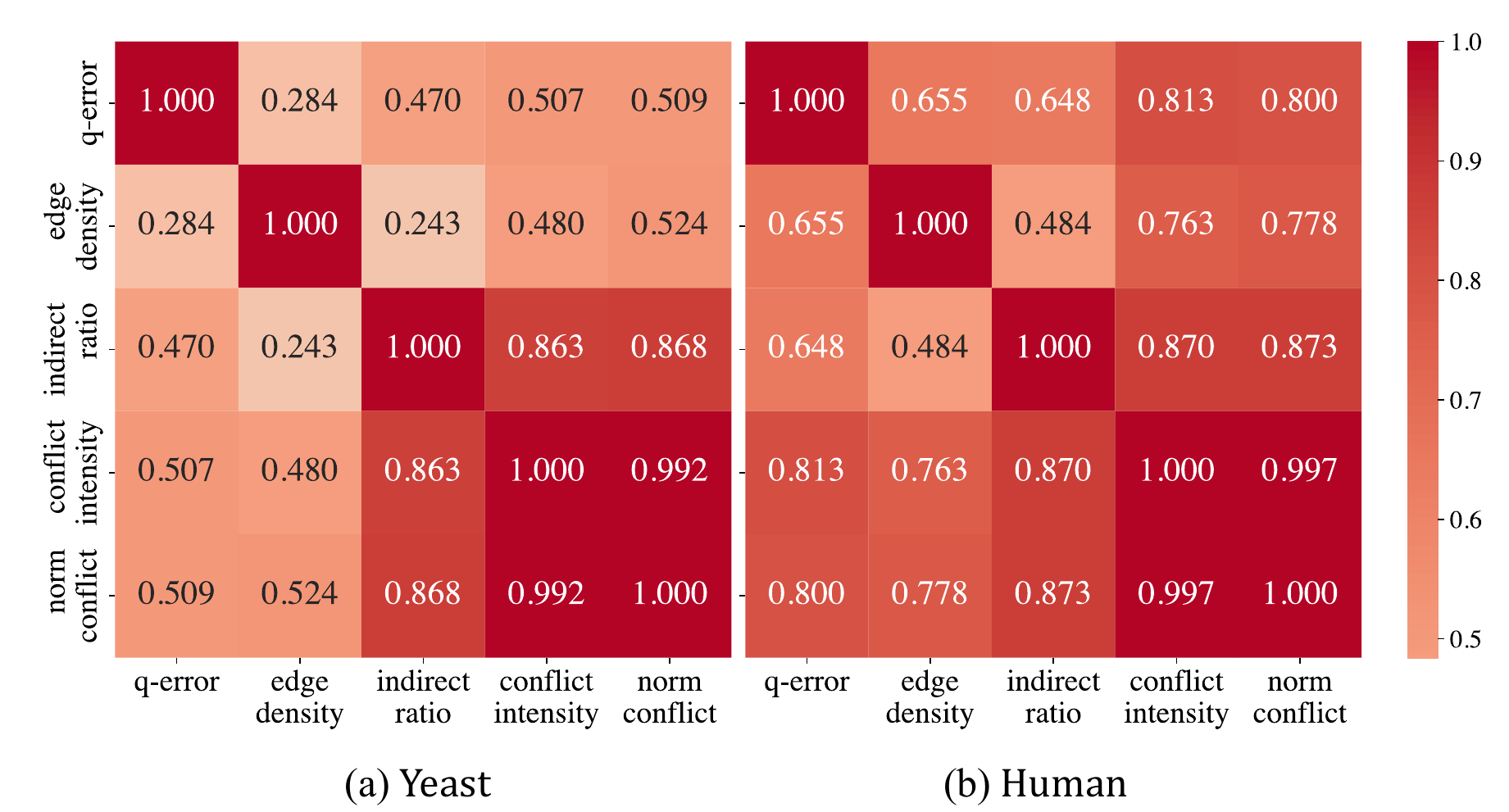}
    \caption{Correlation heatmap.}
    \label{fig:EC_T1}
\end{figure}

\vspace{1mm}
\noindent \textbf{Impact of Sequence Order.}
From the perspective of the FAQ framework~\cite{faq}, our constrained sequence defines a variable elimination order for the underlying subquery. Consequently, the resulting matrix chain multiplication can be viewed as a specialized execution plan for the associated sum-product query.
We evaluate $6$ constrained sequence ordering strategies, namely (1) Random ordering; (2) Min-Candidates, prioritizing query vertices with the smallest candidate set; (3) Max-Degree, prioritizing vertices with the highest degree; (4) Max-Candidates, prioritizing vertices with the largest candidate set; (5) Min-Degree, prioritizing vertices with the lowest degree; and (6) Heuristic, a greedy strategy that minimizes nested or intersecting indirect edges within the sequence. The average \qerror of each $W(u, v)$ over queries from \Yeast and \Human for each strategy is presented in \autoref{tab:sequence_order}. The results show that Heuristic performs the best, which aligns with our previous analysis, i.e., fewer conflicting indirect edges bring less accumulated error. 
In addition, we observe that different ordering strategies consume comparable query time as expected.
This is because the number of vertices and edges in the sequence remains unchanged across different orders.

\vspace{1mm}
\noindent \textbf{Analysis on Type-II module.}
For the Type-II module, when computing the LCA Ratio, we only account for the number of valid combinations in the LCA-endpoints triangle. For the two virtual edges between the LCA and the endpoints, we consider only reachability and ignore the actual number of paths in the candidate space for efficiency. In other words, regardless of how many tree paths connect an endpoint candidate to an LCA candidate, we treat them as a single path. 
Due to the infeasibility of exact algorithms, we evaluate the accuracy of the LCA Ratio for each LCA candidate. In particular, for each LCA candidate, we compute the \textit{variance} in the number of paths from endpoint candidates for each non-tree edge to that LCA candidate, and analyze its correlation with estimation accuracy. 
The experiment results report that Spearman's rho reaches $0.65$ and $0.42$ on \Yeast and \Human, respectively, both showing significant p-values, which confirms our hypothesis.

\vspace{1mm}
\noindent \textbf{Ablative Analysis.}
We conduct ablation studies on $3$ modules for handling Type-I non-tree edge, Type-II non-tree edge, and Iso constraints. 
All variants follow the framework of candidate tree counting. (-) Type-I denotes the model without Type-I non-tree constraint computation, while (-) Type-II and (-) Iso have similar meanings.
Iso-only removes both Type-I and Type-II non-tree constraint computations. 
We conduct experiments on $4$ representative datasets, namely \Yeast, \Human, \Wordnet, and \EU, with different characteristics on average degree and label distributions.

\textit{Accuracy}.
It is reported on \autoref{tab:ablation} that \asc achieves the best accuracy performance with $97.69$ q-error, 
whereas others perform much worse. For example, Iso-only generates an extremely high q-error of $2.1*10^{16}$. 
Surprisingly, the model without the Iso constraint module (i.e., (-) Iso) performs fairly well and even outperforms \fastest.
This implies that our Type-I and Type-II modules can effectively capture structural constraints of the graph. In contrast, without Type-I and Type-II modules, the tree-homomorphism-based Iso module alone cannot achieve satisfactory accuracy.

\begin{table}[t]
    \centering
        \caption{The impact of different orders.}
    \resizebox{\linewidth}{!}{
        \begin{tabular}{l | r r r r r r }
        \toprule
        Order & Random & Min-Candidates & Max-Degree & Max-Candidates & Min-Degree & Heuristic \\ \midrule
        \Yeast & 1.23 & 1.24 & 1.44 & 1.47 & 1.27 & \textbf{1.21} \\ 
        \Human & 2.23 & 4.64 & 9.28 & 9.61 & 2.88 & \textbf{2.13} \\
         \bottomrule
        \end{tabular}    
    }
    \label{tab:sequence_order}
\end{table}

\begin{table}[t]
    \centering
        \caption{Effect of different modules.}
    \resizebox{\linewidth}{!}{
        \begin{tabular}{l | r r r r r r}
        \toprule
        Module & (-) Type-I & (-) Type-II & (-) Iso & Iso-only & \fastest & \asc  \\ \midrule
        \qerror & $6.8*10^9$ & $7.3*10^{15}$ & $1.2*10^6$ & $2.1*10^{16}$ & $7.5*10^6$ & $97.69$ \\ 
        Time (s) & 1.36 & 1.11 & 1.30 & 1.09 & 3.37 & 1.37 \\
         \bottomrule
        \end{tabular}    
    }
    \label{tab:ablation}
\end{table}

\textit{Query Time}. As reported in \autoref{tab:ablation}, all variants consume almost the same amount of query time. This indicates that removing a specific module does not bring a clear efficiency improvement. The minor time differences are mainly attributed to the distribution of constraints in the query graph.

\vspace{1mm}
\noindent \textbf{Constraint Analysis.}
\begin{figure*}[t]
    \centering
    \begin{minipage}{0.75\linewidth}
        \centering
        \includegraphics[width=\linewidth]{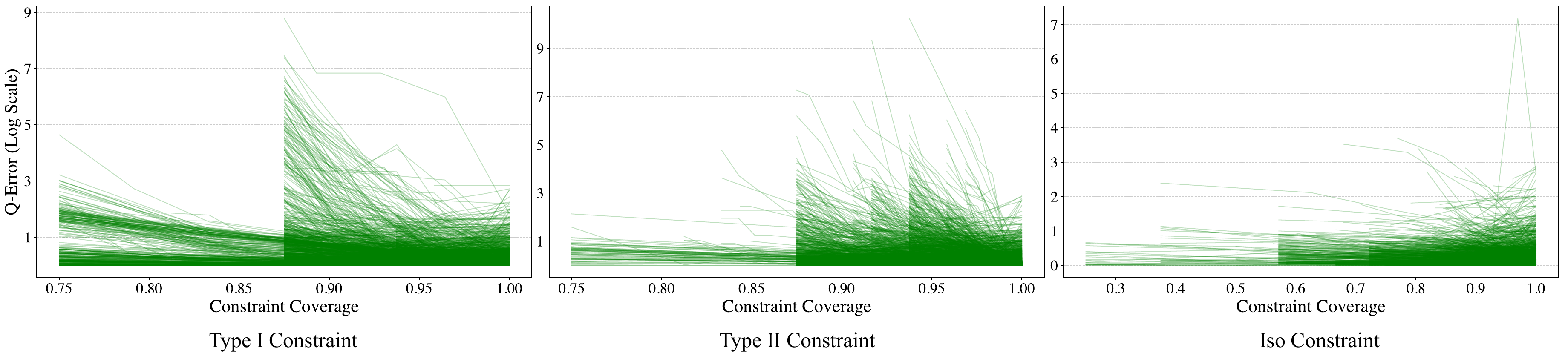} 
        \caption{Accuracy against constraint coverage.}
        \label{fig:parameter}
    \end{minipage}
    \hspace{-4mm}
    \begin{minipage}{0.25\linewidth}
    \vspace{-14mm}
        \centering
        \captionof{table}{Time growth.\label{tab:time_growth}}
        \vspace{-2mm}
        \resizebox{0.9\linewidth}{!}{
            \begin{tabular}{l|r r}
            \toprule
            Type & Cov. = 0.0 & Cov. = 1.0 \\ \midrule
            Type-I & 0.82 & 0.83 \\
            Type-II & 0.83 & 0.83 \\
            Iso & 0.81 & 0.83 \\ \bottomrule
            \end{tabular} 
        }
    \end{minipage}

\end{figure*}
To study how the constraint intensity of each module affects the counting performance, we vary the constraint intensity within the candidate tree counting framework. Specifically, for each module, we control the proportion of its constraints involved in the counting across $5$ incremental levels: $0\%, 25\%, 50\%, 75\%$, and $100\%$. For Type-I non-tree constraints, which arise in the non-tree edges of each constrained sequence in the candidate tree, an $x\%$ coverage rate signifies that only $x\%$ of these non-tree edges are considered during computation, while the left $2$ modules are kept intact. Type-II constraints involve non-tree edges with different parent vertices, and we group these edges by the LCA vertex of their endpoints. Similarly, we restrict the number of non-tree edges used in LCA Ratio computation within each LCA group by utilizing only $x\%$ of the edges, and the remaining $2$ modules are kept intact. For iso constraints, which target groups of same-label vertices that may violate injectivity, we randomly select $x\%$ of the vertices in the group for constraint computation.

We conduct experiments on \EU and \Human, which have high average degrees, and \Yeast, which has a moderate average degree. 
Although we preset $5$ coverage ratio levels, the actual constraint coverage often deviates significantly. Within the candidate tree counting framework, not every parent vertex in the query tree poses Type-I non-tree constraints in its child vertices. For the parent vertices without Type-I constraints, the weight updates degrade to \autoref{eq:tree_counting}, resulting in exact computation. If only one parent vertex involves Type-I non-tree constraints while others remain exact, the overall accuracy impact remains limited, even if that specific parent ignores its constraints entirely. Consequently, evaluating performance based on preset coverage yields erratic, irregular fluctuations. To ensure a rigorous analysis, we utilize actual constraint coverage instead. Each vertex in the query tree is assigned a local constraint coverage, representing the actual coverage ratio among its child vertices. For vertices that have no children or lack non-tree edges between their children, the local coverage is set to $1.0$. For all other vertices, we randomly select non-tree edges based on the preset coverage and calculate the resulting actual coverage. The global actual coverage for each query is then defined as the average of the local coverage across all its vertices. Similarly, we calculate the actual Type-II constraint coverage for each vertex. For Iso constraints, vertices are grouped by their labels. Groups that do not violate injectivity are assigned a local coverage of $1.0$. For the remaining groups, the local coverage is derived from the vertices retention rate within the connected subgraph used for sampling. The overall actual coverage is then determined by averaging the coverage across all groups.

\autoref{fig:parameter} illustrates the accuracy performance of all queries across the three datasets relative to the varying coverage of the three constraint types. The actual constraint coverage is determined by both the query graph structure and the preset ratios. Overall, for all three constraint types, the accuracy for most queries improves as the coverage increases. However, a small number of queries do not achieve optimal accuracy at $100\%$ coverage. For the Type-I constraint, the computation is not perfectly exact. For constrained sequences containing nested cycles, we apply penalties to potential invalid weights based on pre-calculated reachable weight ratios. Depending on the actual matching states within the Cartesian space, the penalty can be either excessive or insufficient. In such cases, removing certain non-tree edges may inadvertently mitigate these over- or under-penalizations, thereby leading to improved accuracy. Type-II constraint computation is performed at the LCA-group level. Within a single LCA group, the LCA ratio only considers whether a candidate of a non-tree edge endpoint is connected to the LCA candidate, while actually multiple valid tree paths may exist between them. Since our estimation does not capture the variation in the number of such paths, removing certain non-tree edges can occasionally compensate for part of the estimation error. In the Iso constraint sampling process, certain regions within the query subgraphs and their corresponding candidate space can be particularly difficult to sample. For instance, the local area of the query subgraph may exhibit high symmetry, resulting in a high density of homomorphisms but sparse isomorphisms within the candidate space. By limiting the constraint coverage, this symmetry may be broken, which can lead to a slight improvement in accuracy.

In summary, while limiting constraint coverage may occasionally offset inherent model errors in some queries, a $100\%$ coverage rate remains the optimal choice for most queries. Furthermore, as shown in Table~\ref{tab:time_growth}, increasing the coverage from $0\%$ to $100\%$ results in negligible changes to the average query processing time.

\subsection{Discussion}
\asc achieves high accuracy in homomorphism counting as shown in \Cref{tab:homo_acc}, benefiting from its precise handling of structural constraints, especially the non-tree ones. While it also performs well for isomorphism counting, the iso ratio module remains the primary error source, as estimating the isomorphism ratio via sampling is inherently less accurate than the algebraic counting. We see two directions to improve it. First, the current sampling could be replaced by more robust alternatives, since sampling still faces inherent limitations: global sampling (e.g., in \fastest and \alley) often becomes inaccurate or fails in vast search spaces, and although our local sampling alleviates this by skipping query regions that cannot violate injectivity and sampling only over the conflict-prone regions, it does not fundamentally escape the accuracy limitations of sampling. Second, although the structural constraints, including non-tree edges, are already incorporated into the counting, their effect is currently confined to the weight accumulation in the counting stage. Another promising direction is thus to let such constraints act directly on the Cartesian-product space underlying the counting, pruning its enumerable states rather than only correcting the accumulated weights as in the current design.

\section{Related Work}

\noindent \textbf{Subgraph Matching}.
Subgraph matching research has been dominated by the filtering-ordering-enumeration paradigm. Recent studies have enhanced this framework through efficient filtering \cite{GraphQL, TurboIso, CECI, CFLMatch, DAF, Fastest, flowsc}, optimized ordering strategies \cite{CFLMatch, VF2, quicksi, RISubgraphMatching}, and advanced enumeration techniques \cite{RapidExperiment, 2024survey}. Existing enumeration techniques primarily fall into two categories: exploration-based and decomposition-based methods. Within exploration-based approaches, performance is often improved via search space pruning \cite{gup, DAF, bee}, intermediate result reuse \cite{suff, VEQ, kim2023fast, bsx, lu2025accelerating}, and recursion reduction \cite{circinus, bsx, sun2019efficient, calig}. Alternatively, decomposition-based methods \cite{aberger2017emptyheaded, mhedhbi12optimizing, rapidmatch, lai2017scalable, lai2016scalable, qiao2017subgraph, li2025subgraph} partition the query graph and employ join techniques, such as binary joins \cite{lai2017scalable, lai2016scalable, sun2012efficient} and worst-case optimal joins (WCOJ) \cite{WCOJoin, aberger2017emptyheaded, mhedhbi12optimizing, rapidmatch}.

\vspace{1mm}
\noindent \textbf{Subgraph Counting}.
Subgraph counting primarily employs three approximate paradigms. Summarization-based methods \cite{CSET, SumRDF, 2014combinatorial, color} decompose queries and aggregate substructure counts, yet often struggle with accuracy due to independence assumptions. Sampling-based approaches \cite{IMPR, MOTIVO, WanderJoin, Fastest} are widely adopted but face inherent sampling failures, despite recent efforts in search-space compression and strategy optimization. Finally, learning-based methods \cite{flowsc, LSS, NeurSC,alss} leverage GNNs to capture graph interactions, but their effectiveness is limited by the complex modeling between graph features and counts, high training overhead, and heavy dependence on ground-truth supervision.

\section{Conclusion}
In this paper, we study the problem of subgraph counting. We propose an algebraic counting method based on the candidate tree-based counting framework.
Using fast matrix computation operations, we can obtain accurate subgraph homomorphism counting with high efficiency. Based on this, 
we obtain the final subgraph isomorphism count via a local sampling strategy. Extensive experiments demonstrate the superior performance of our proposals.

\section*{Acknowledgements}
\noindent This work was supported by China Scholarship Council (20230844
0228), 
National Natural Science Foundation of China (62572137), Guangdong Basic and Applied Basic Research Foundation 
(2025A15\allowbreak15011716), Major Key Project of PCL (PCL2024A05 and PCL2025A16), 
Australian Research Council Centre of Excellence for Mathematical Modelling of Cellular Systems (CE230100001), Australian Research Council Discovery Project (DP260100689), 
Australian Research Council Discovery Early Career Researcher Awards (DE250100226),
and National Natural Science Foundation of China (U2241211).

\newpage
\bibliographystyle{ACM-Reference-Format}
\bibliography{bib}
\appendix

\end{document}